\documentclass{article}

\usepackage{PRIMEarxiv}

\usepackage[utf8]{inputenc} 
\usepackage[T1]{fontenc}    
\usepackage{hyperref}       
\usepackage{url}            
\usepackage{booktabs}       
\usepackage{amsfonts}       
\usepackage{nicefrac}       
\usepackage{microtype}      
\usepackage{lipsum}
\usepackage{fancyhdr}       
\usepackage{graphicx}       
\usepackage{xcolor}         
\graphicspath{{media/}}     
\usepackage{subcaption}

\usepackage{verbatim}

\usepackage{amsmath}
\usepackage{amsthm}
\usepackage{natbib}
\usepackage{dsfont}

\newtheorem{theorem}{Theorem}
\newtheorem{corollary}{Corollary}
\newtheorem{remark}{Remark}
\newtheorem{definition}{Definition}
\newtheorem{algo}{Algorithm}
\newtheorem{lemma}{Lemma}

\providecommand{\R}{}
\renewcommand{\R}{\mathds{R}}

\newcommand\cS{{\mathcal S}}
\newcommand\cN{{\mathcal N}}
\newcommand\cM{{\mathcal M}}
\newcommand\cH{{\mathcal H}}
\newcommand\cO{{\mathcal O}}

\newcommand{\tends}{\rightarrow}
\newcommand{\bran}[1]{\left\{#1\right\}}

\newcommand{\ind}{\mathds{1}}
\newcommand{\PPP}{{\text{PPP}}}

\renewcommand{\d}{\mathrm d}

\newcommand{\expit}{{\text{expit}}}

\newcommand{\mat}[1]{\boldsymbol{\mathrm{#1}}}
\newcommand{\GP}{{\text{GP}}}

\newcommand{\iif}{\Leftrightarrow}
\newcommand{\indep}{\ensuremath{\stackrel{\text{ind.}}{\sim}}}

\DeclareMathOperator*{\argmax}{arg\,max}

\title{Bayesian Analysis of Sigmoidal Gaussian Cox Processes via Data Augmentation
}

\author{
  Renaud Alie, David A. Stephens \\
  Department of Mathematics and Statistics \\
  McGill University \\
  Montreal\\
   \And
  Alexandra M. Schmidt \\
  Department of Epidemiology, Biostatistics and Occupational Health\\
  McGill University \\
  Montreal\\
}

\begin{document}
\maketitle

\begin{abstract}
Many models for point process data are defined through a thinning procedure where locations of a base process (often Poisson) are either kept (observed) or discarded (thinned). In this paper, we go back to the fundamentals of the distribution theory for point processes to establish a link between the base thinning mechanism and the joint density of thinned and observed locations in any of such models. In practice, the marginal model of observed points is often intractable, but thinned locations can be instantiated from their conditional distribution and typical data augmentation schemes can be employed to circumvent this problem. Such approaches have been employed in the recent literature, but some inconsistencies have been introduced across the different publications. We concentrate on an example: the so-called sigmoidal Gaussian Cox process. We apply our approach to resolve contradicting viewpoints in the data augmentation step of the inference procedures therein. We also provide a multitype extension to this process and conduct Bayesian inference on data consisting of positions of two different species of trees in Lansing Woods, Michigan. The emphasis is put on intertype dependence modeling with Bayesian uncertainty quantification.
\end{abstract}


\section{Introduction} \label{Intro}

Spatial point processes describe the random behavior of point configurations in space. Several real-life phenomena have motivated the development of point process theory such as the position of stars in a galaxy \citep{babu1996spatial}, the geographical position of trees \citep{wolpert1998poisson} and the space-time locations of earthquakes \citep{ogata1998space}. The simplest model for point configurations is the Poisson point process (PPP) which is typically characterized through properties of its count function \citep{kingman1992poisson,moller2003statistical}. PPPs come in finite and non-finite varieties, but only the finite case admits a constructive definition and a density. Cox models \citep{cox1955some} are hierarchical Poisson models with a stochastic intensity function. Even in the finite case, their density is usually not available in closed form. On the other hand, finite Markov point processes are described by properties of their density function \citep{moller2003statistical,van2019theory}, which are traditionally taken with respect to the distribution of a PPP with unit intensity.

Other point process models are defined through a generative thinning procedure where only a subset of a base point process is observed. Examples include the three types of Mat\'ern processes \citep{matern1960spatial} where a deterministic thinning rule is applied to a base PPP. Mat\'ern thinning prevents two points from lying within a predetermined distance of one another. Another important example is the thinning procedure of \citet{lewis1979simulation} to generate non-homogeneous PPPs. This one is probabilistic and operates independently across locations of the base process. It is usually the case that the marginal density of observed locations is not available in tractable form since one would need to integrate out the base point process along with the thinning procedure over every compatible configuration.

We focus on a Bayesian formulation of statistical inference for point process models, but some of the results we present also apply to other approaches where data augmentation is involved. Consider an observed point process $Y$ which we postulate to be distributed according to a density $f(Y|\theta)$ (wrt to some measure $\mu$) indexed by a finite-dimensional parameter $\theta$. We assign a prior density $\pi(\theta)$ (wrt to the Lebesgue measure $\lambda$) to the model parameters. This is a typical Bayesian specification where $f(Y|\theta)\pi(\theta)$ is understood as the joint density of $(Y,\theta)$ wrt to the product of $\mu$ and $\lambda$ measures. In this context, one can use the observed $y$ and interpret $f(y|\theta)\pi(\theta)$ as, up to a normalizing constant, the density wrt to $\lambda$ of the regular probability distribution \citep[Section 4.1.3]{durrett2019probability} of $\theta$ given $Y$.  The situation is trickier, however, when considering data augmentation to circumvent some intractability in the model $f(Y|\theta)$. In such a case, a joint density $f(X,Y|\theta)$ is needed along with the prior, where $Y$ represents the observed data and $X$ the unobserved quantity.  This joint model needs to be valid in the sense that it is a density wrt to some product measure $\nu \times \mu$ over the product of $X$ and $Y$ spaces. Proportionality arguments can then be used to interpret $f(X|Y,\theta) \propto f(X,Y|\theta)\pi(\theta)$ as the full conditional density (wrt to $\nu$) of $X$ given $Y$ and $\theta$. This fact is of direct importance in many computational approaches to the Bayesian analysis of point process data and models, which often utilize augmentation as a strategy.


We go back to the fundamentals of probability theory for point processes to resolve this obstacle and provide a construction of the posterior distribution that can be used in a general setting.  Considering everything else fixed, any statement about the conditional distribution of thinned locations $X$ given the observed $Y$ should be obtainable from their joint density, provided they admit such. We present a colouring theorem that characterizes the joint density of thinned and observed locations for any thinning procedure applied to any point process that admits a density.  This paper focuses on delicate aspects concerning computations for one class of point process models, the sigmoidal Gaussian Cox process (SGCP) \citep{adams2009tractable}. We also examine, in Appendix \ref{MT3sect}, alternate derivations for some important results from the literature about the Mat\'ern type III process.  We illustrate how one specific result, the colouring theorem, streamlines most of the measure-theoretic details involved. Moreover, it provides a robust and unequivocal framework to state and verify any claim about joint, marginal and conditional distributions in models based on thinning procedures.  The colouring theorem is not limited to binary categories (thinned and observed). It can be employed to derive the joint density of any number of point processes based on a categorical marking of a base process. This provides a powerful tool to construct multitype point processes with interesting dependence structure while retaining a tractable density. We use this result to provide a multitype extension of the SGCP along with the appropriate inference scheme that preserves the advantages of the original model: it does not involve the discretization or truncation approximations that are common to similar methods. We show how our multitype model and algorithms can be useful in investigating the interaction between different species of trees \citep{gerrard1969new} from Lansing, Michigan: see Figure \ref{fig:mapHick}.
\begin{figure}[t!]
    \centering
    \includegraphics[height = 7cm]{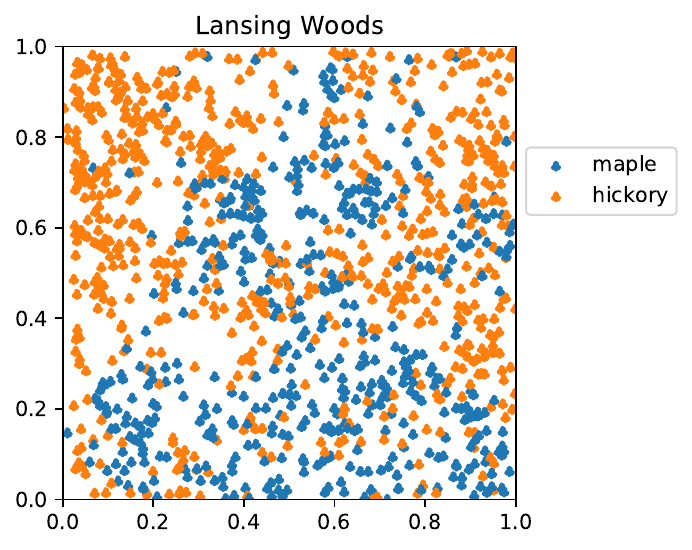}
    \caption{The location of maples and hickories in Lansing Woods, Michigan.}
    \label{fig:mapHick}
\end{figure}


In the last decade or so, there has been a growing literature on non-homogeneous PPP modeling. Originally, \citet{adams2009tractable} proposed a data augmentation scheme to conduct inference about the intensity function of an observed point process $Y$. While the marginal PPP likelihood involves the integral of the intensity function, the complete data joint density $f(X,Y|\theta)$ of thinned ($X$) and observed ($Y$) locations is itself tractable. Imaginary thinned locations $X$ are instantiated at each step of the Markov chain Monte Carlo (MCMC) algorithm according to their full conditional distribution in a Metropolis-Hastings type of algorithm using birth-death-move proposals. A space-time extension of this model is proposed in \citet{gonccalves2018exact}. The authors derive an alternative way to perform the data augmentation step by simulating the thinned locations exactly from their full conditional. As noted in their paper, their approach is incompatible with that of \citet{adams2009tractable}. In \citet{rao2017bayesian}, the authors use the same type of data augmentation but in the case of Mat\'ern type III thinning of a non-homogeneous PPP. At the first step of the hierarchy is the sigmoidal Gaussian Cox process (SGCP) model of \citet{adams2009tractable}. The authors present yet another new procedure to instantiate the thinned point process from its exact full conditional. Even accounting for all the particularities associated with each model, the data augmentation step to be performed is the same in each of \citet{adams2009tractable}, \citet{rao2017bayesian} and \citet{gonccalves2018exact}. However, it can be shown that they each simulate from a different distribution. It is worth noting that the approach of \citet{gonccalves2018exact} has since been amended in a corrigendum \citep{gonccalves2023corrigendum}. The authors adopt a new update strategy for the thinned point process, this one based on a retrospective sampling argument. Under a new united formalism, our work provides a critical assessment of the four methods mentioned above.

Other comparable models based on latent random measures include the Poisson-gamma model of \citet{wolpert1998poisson} which has been extended to multitype modeling in \citet{kang2014bayesian}. The other notable example is the method of \citet{kottas2007bayesian} in which the intensity measure is described by a scaled Dirichlet process mixture.  This one has, hitherto, not been generalized to joint modeling of point patterns. Both methods include partial sum approximations of infinite mixtures. In contrast, our multitype adaptation of the SGCP preserves the advantages of the unitype version: it is exact up to Monte Carlo error. In the next section, we start by reviewing some basic notions of point process distribution theory. We describe and justify the type of dominating measure we shall use throughout as it will be necessary to discuss point processes in general spaces.

\section{Point process modelling} \label{ppsect}

\subsection{Point Process Densities}\label{Dens}

The random locations of a point process take their values in some space $\cS$. Usually, $\cS$ corresponds to a real-life physical space: we have that $\cS$ is a bounded subset of $\R^d$ with $d=1,2$ or $3$. Whenever $d>1$, the locations exhibit no natural ordering. In this sense, it is usual to describe a point process $X$ as a random countable subset of $\cS$ \citep{moller2003statistical}.  We will concentrate on \textit{finite} point processes (FPP) that have a finite number of locations with probability one. Therefore, we can assume that $X$ is a random finite subset of $\cS$ taking its values in
$
\cN_{\text{f}}(\cS)= \bran{S \subset \cS: N(S) < \infty},
$
where we use $N(S)$ to denote the cardinality of a set $S$.


PPPs are parametrized by a positive intensity function $\lambda:\cS \tends [0,\infty)$. We write $X\sim \PPP(\cS,\lambda(\cdot))$. In the finite case, we can construct such a process as follows:
\begin{enumerate}
    \item The total number of points $N(X)$ has the Poisson($\Lambda(\cS)$) distribution, where $\Lambda(\cS) = \int_\cS\lambda(s) \d s$.
    \item Conditional on $N(X)$, the points are independently scattered over $\cS$ according to the probability density $\lambda(\cdot)/\Lambda(\cS)$.
\end{enumerate}
The process outlined above assumes the integral of $\lambda(\cdot)$ over $\cS$ is finite. For a measurable subset $F$ of $\cN_{\text{f}}(\cS)$, we can write the distribution of $X$ as
\begin{align}
   P(X \in F) = \sum_{n \geq 0} \frac{\exp(-\Lambda(\cS))}{n!}\int_{\cS^n} \ind_F(\{x_1,x_2,\dots,x_n\}) \prod_{i=1}^n \lambda(x_i) \d x_1 \d x_2 \dots \d x_n. \label{distPois}
\end{align}
For completeness, the integral over $\cS^0$ should be understood as evaluating the inside function at the empty set $\emptyset$.

Applying the monotone class theorem \citep[Section 5.2]{durrett2019probability}, we can extend the previous result from indicator functions to any positive and measurable function $h:\cN_{\text{f}}(\cS) \tends [0,\infty)$:
\begin{align}
E[h(X)] = \sum_{n \geq 0} \frac{\exp(-\Lambda(\cS))}{n!}\int_{\cS^n} h(\{x_1,x_2,\dots,x_n\}) \prod_{i=1}^n \lambda(x_i) \d x_1 \d x_2 \dots \d x_n. \label{expPois}
\end{align}

The unit rate PPP is defined by taking $\lambda(\cdot)$ to be identically 1 everywhere on $\cS$. If we define $Y\sim \PPP(\cS,1)$, then $Y$ is finite iff the space $\cS$ is bounded ($|\cS| < \infty$). In such a case, we can derive the density of $X$ with respect to the distribution of $Y$ as
$$
f(\{x_1,x_2,\dots,x_n\}) = \exp\Big(|\cS| - \Lambda(\cS) \Big)\prod_{i=1}^n \lambda(x_i)
$$
by careful manipulations of equations \eqref{distPois} and \eqref{expPois}.

Taking the density of a point process with respect to the $\PPP(\cS,1)$ distribution is a common approach in the literature. We argue in the following that being restricted to a bounded domain $\cS$ can easily become hindering especially when considering marked point processes.

\subsection{Counting-Scattering Measure}\label{csDens}

Taking the density of a FPP wrt to the unit rate PPP is analogous to taking the Radon-Nikodym derivative of an absolutely continuous distribution with respect to the uniform distribution, which is only valid if the support is bounded. In the unbounded $\cS$ case, one could define the density wrt to another Poisson process with intensity $\rho(\cdot)$ chosen so that it is finite, i.e. $\int_\cS \rho(s)\d s < \infty$. This is analogous to taking the derivative of a continuous distribution with respect to an arbitrary distribution.  However, the density of a FPP need not necessarily be taken wrt to a probability distribution. In the case of continuous random variables, we generally understand the term {\it density} as the derivative wrt to the Lebesgue measure which is not a finite measure. Upon inspecting equation \eqref{distPois}, we can see that the distribution of any finite PPP (regardless of whether $\cS$ is bounded or not) is absolutely continuous wrt to the $\sigma$-finite measure
\begin{align}
   \mu_{\text f}(F) =   \sum_{n \geq 0} \int_{\cS^n} \ind_F(\{x_1,x_2,\dots,x_n\})  \d x_1 \d x_2 \dots \d x_n. \label{CSmeas}
\end{align}

Moreover, any FPP defined as in \citet[Section 5.3]{daley2003introduction} from a probability mass function $p_n$ ($\sum_{n \geq 0} p_n = 1$) on the total number of points and a family of symmetric densities $\{\pi_n(\cdot), n\geq 1\}$ (each one over $\cS^n$) has density
\begin{align}
    f(\{x_1,x_2,\dots,x_n\}) = p_n \pi_n(\{x_1,x_2,\dots,x_n\}) \label{FPPdens}
\end{align}
wrt $\mu_{\text f}$. We refer to the process of assigning locations conditional on the number of points as \textit{scattering} and call $\mu_{\text f}$ the \textit{counting-scattering} measure. The density wrt to $\mu_{\text f}$ is arguably more informative about the process than its counterpart taken wrt to the unit PPP because it relates directly to a simulation procedure. The density wrt \eqref{CSmeas} of the $\PPP(\cS,\lambda(\cdot))$ distribution has the form
\begin{align}
  f(\{x_1,x_2,\dots,x_n\}) = \frac{\exp(-\Lambda(\cS))}{n!} \prod_{i=1}^n \lambda(x_i), \label{PPPCSdens}
\end{align}
provided it is finite.

The approach outlined above is inspired by Janossy measures \citep{janossy1950absorption} and how they describe point process distributions. More recently, both \citet{rao2017bayesian} and \citet{zhang2017independence} defined a measure akin to \eqref{CSmeas}. Distinctively though, both papers consider FPPs to be ordered, albeit arbitrarily so. There is nothing inherently contradictory with this approach, but considering FPPs as sets is more in line with the established theory. Finally, the counting-scattering measure approach is also more easily adapted to more general spaces $\cS$. The measure $\mu_{\text f}$ as expressed in \eqref{CSmeas} is a hybrid between the counting and Lebesgue measures, but the scattering part of the distribution could use any $\sigma$-finite measure over a space $\cS$ that may not be Euclidian. This will be useful in the next section when we consider FPPs with marks in general spaces.

\subsection{Markings}\label{Marks}

Marks, unlike locations, do not generally lie in some physical space and therefore are not necessarily restricted to bounded spaces. Moreover, one could easily envision marks supported on a non-Euclidian space such as one with a discrete component. The unit rate PPP can be defined on more general spaces than subsets of $\R^d$, but this is not common in the statistical literature.  The counting-scattering measure is readily adapted to general spaces, however.

For finite point processes, we generally consider locations to precede marks. The latter are scattered conditionally on the former through a family of models
\begin{align}
    \pi_n(m_1,\dots,m_n|s_1,\dots,s_n),\quad \text{for } n=1,2,3,\dots \label{condScat}
\end{align}
The conditional densities above should be invariant under permutations of indices $\{1,2,\dots,n\}$ to preserve the natural symmetry of point processes.  Fundamentally, the marks of a finite point process are themselves a FPP: a random number of locations in some mark space $\cM$. For this reason, it is usual to consider the set of locations with corresponding marks as a single point process over the product space $\cS \times \cM$. Along with location scattering distributions, the conditional densities in \eqref{condScat} specify the scattering models of this augmented point process:
\begin{align}
\pi_n(\{(s_1,m_1),\dots,(s_n,m_n)\}) = \pi_n(m_1,\dots,m_n|s_1,\dots,s_n) \pi_n(s_1,\dots,s_n) \label{augScat},
\end{align}
for any $n\geq 1$. From this point on, we assume the densities described in $\eqref{augScat}$ to be respectively taken wrt to the product measure $(\mu_\cS\times \mu_\cM)^n$ over measurable subsets of $(\cS \times \cM)^n$, where $\mu_\cS$ (resp. $\mu_\cM$) is the dominating measure over the location (resp. mark) space.

As indicated above, there is no real distinction between points and marks in the product space representation. The two are equally important components of the same specification: a random finite subset of the product space $\cS \times \cM$. As indicated in Section \ref{csDens}, we can write the density of this FPP as $p_n \pi_n(\{(s_1,m_1),\dots,(s_n,m_n)\})$ without specifying whether a standard rate PPP is finite or even defined on such a space. This density is taken wrt to the counting-scattering measure defined by
\begin{align*}
\mu_{\text f} (F) = \sum_{n\geq 0} \int_{(\cS \times \cM)^n} \ind_F(\{(s_1,m_1),\dots,(s_n,m_n)\}) \mu_\cS(\d s_1)\mu_\cM(\d m_1) \dots \mu_\cS(\d s_n)\mu_\cM(\d m_n)
\end{align*}
for any measurable subset $F$ of $N_{\text f}(\cS \times \cM)$. How to construct a measure space over $N_{\text f}(\cS \times \cM)$ is thoroughly discussed in \citet{zhang2017independence}. Essentially, one can define it from the measure spaces of $\mu_\cS$ and $\mu_\cM$ which are assumed to be provided. Marked point processes are akin to regular point processes except that they can take their values in general spaces rather than being confined to bounded subsets of $\R^d$ where we typically observe locations. For this reason, it is helpful to avoid the restrictions associated with using the unit rate PPP as a dominating measure.

\subsection{Colourings}\label{Col}

We define a colouring as a particular type of marking: one with a discrete and finite mark space $\cM = \{0,1,2,\dots,K\}$. Although it is essentially the simplest type of marking one could imagine, colouring procedures are an interesting tool to construct multitype point processes. Indeed, by considering each colour as its own point configuration, we can study the joint, marginal and conditional distributions among the various types. A thinning procedure is a colouring with mark space $\cM =$ \{thinned, observed\} $\equiv\{0,1\}$. Examples of thinnings include the procedure of \citet{lewis1979simulation} to simulate non-homogeneous PPPs or the three types of Mat\'ern repulsive processes \citep{matern1960spatial}.

Next, we present an important conceptual result that we term a colouring theorem in honor of the homonymous result in \citet{kingman1992poisson}. The principle of a colouring theorem is to describe the marginal and joint structure of a multitype point process created by splitting a base point process into colours. The original version tells us that if the points of a PPP are coloured independently of one another, then the point processes corresponding to each type are also Poisson and they are mutually independent. The result below is more general in that it applies to any FPP with a density and arbitrary colouring mechanisms.

\begin{theorem}[Colouring theorem]\label{colThm}
Consider a FPP over the product space $\cS \times \{0,1,\dots,K\}$ with density
$$f(\{(s_1,c_1),\dots,(s_n,c_n)\}) = p_n \pi_n(\{(s_1,c_1),\dots,(s_n,c_n)\}).$$
The dominating measure over the discrete mark space is the counting measure, hence we have
$$
\sum_{n\geq 0} \int_{\cS^n}\sum_{\mat c \in \{0,1,\dots,K\}^n} f(\{(s_1,c_1),\dots,(s_n,c_n)\}) \mu(\d s_1) \dots \mu(\d s_n) = 1
$$
for some arbitrary measure $\mu$ over subsets of $\cS$. Let $\mu_{\mathrm{f}}$ denote the counting-scattering measure corresponding to $\mu$ for FPPs over $\cS$, that is
$$
\mu_{\mathrm{f}}(F) = \sum_{n \geq 0} \int_{\cS^n} \ind_F (\{s_1,\dots,s_n\})\mu(\d s_1) \dots \mu(\d s_n), \quad \text{for measurable } F \subseteq N_{\mathrm{f}}(\cS).
$$
Define the point processes $X_0,X_1,\dots,X_K$ as the locations in $\cS$ with respective colours in $\{0,1,\dots,K\}$. Let $S_0,S_1,\dots,S_K$ be finite subsets of $\cS$ with respective sizes $n_0,n_1,\dots,n_K$. Then the joint density of $X_0,X_1,\dots,X_K$ wrt to the product measure $\mu_{\mathrm{f}}^{K+1}$ has the form
\begin{align}
    f(S_0,S_1,\dots,S_K) = \binom{n}{n_0,n_1,\dots,n_K} f\bigg(\bigcup_{k=0}^K S_k\times \{k\}\bigg) \label{combTerm}
\end{align}
where $n=\sum_{k=0}^Kn_k$.
\end{theorem}

The proof can be found in Appendix \ref{proColThm}. In the finite case, the original colouring theorem for PPPs is recovered from Theorem \ref{colThm}.

\begin{remark}\label{rem}
It is possible to define a dominating measure such that the density of a FPP is of the form $n!p_n\pi_n(\{s_1,\dots,s_n\})$. This changes the statement in the colouring theorem to
\begin{align}
    f(S_0,S_1,\dots,S_K) =  f\bigg(\bigcup_{k=0}^K S_k\times \{k\}\bigg). \label{equiva}
\end{align}
In this case, the joint density is exactly equal to the density of the discretely marked point process, but it only holds for this particular choice.

\end{remark}

The equivalence described in Remark \ref{rem} is a fortuitous property of a single (and not widespread) dominating measure. We need to emphasize that, in general, the unitype density of the discretely marked point process is not the same thing as the multivariate density. The former is not taken wrt a product measure and therefore cannot be marginalized. This misconception has appeared in the literature and has led to some problematic arguments.

A result equivalent to Theorem \ref{colThm} is stated in \citet[Section 6.6.1]{moller2003statistical}, but no proof is given there. As a dominating measure, the authors use an extension of the standard PPP distribution that handles discrete marks. This arguably makes it more difficult to interpret when compared to our approach based on the counting-scattering measure. The implications of the result are important, yet under-appreciated, as we will discuss in Section \ref{SGCP}. We will mostly limit ourselves to the particular case of independent colourings described in the following result.
\begin{corollary}

Let $f$ be the density of the base point process $X$ over $\cS$. If the points of $X$ are each independently coloured from a PMF $p(c|s)$ for $c = 0,1,\dots,K$ and $s\in\cS$, then the joint density of the point processes $X_0,X_1,\dots,X_K$ corresponding to colours in $\{0,1,\dots,K\}$ has the form
$$
f(S_0,S_1,\dots,S_K) = \binom{n}{n_0,n_1,\dots,n_K} f\bigg(\bigcup_{k=0}^K S_k\bigg) \prod_{k=0}^K\prod_{s\in S_k} p(k|s).
$$
\end{corollary}
In Appendix \ref{MT3sect}, we showcase the generality of Theorem \ref{colThm} by applying it to the Mat\'ern type III process and, doing so, handily recover some important results from this literature.

Theorem \ref{colThm} provides a rigorous framework in which to state and verify statements about the marginal or conditional behavior of thinned and observed locations in such contexts. It is simple enough, yet its implications are important to some of the statistical literature on point processes. In the next section, we showcase how the Gaussian Cox process introduced in \citet{adams2009tractable} and the data augmentation procedure described therein can be better understood through the lens of this colouring theorem. Extensions of this model are proposed in \citet{rao2017bayesian,gonccalves2018exact}, but some of the results across the three publications are inconsistent with one another. We use the framework described above to resolve those issues.

\section{Unitype Sigmoidal Gaussian Cox Process}\label{SGCP}
The sigmoidal Gaussian Cox Process (SGCP) was introduced in \citet{adams2009tractable}. It is a particular type of Cox process: a PPP model with a stochastic intensity function. The SGCP intensity is constructed by mapping a Gaussian process (GP) $g(\cdot)$
to the $[0,1]$ interval through a function $\sigma(\cdot)$. Scaling is handled by a constant $\lambda$. The global model hierarchy is presented below.
\begin{definition}[Sigmoidal Gaussian Cox Process]\label{SGCPdef}
The SGCP is the Cox process defined by
\begin{enumerate}
    \item $g\sim\GP(m(\cdot),C(\cdot,\cdot))$,
    \item $X_1 \sim  \PPP(\lambda \sigma \circ g(\cdot))$.
\end{enumerate}
\end{definition}
The particular form of the mean $m(\cdot)$ and covariance $C(\cdot,\cdot)$ functions or their parametrization will not play a role in this section. The SGCP was originally presented in \citet{adams2009tractable} using the sigmoid $\sigma(\cdot) = \expit(\cdot):= \exp(\cdot)/(1+\exp(\cdot))$ transformation. \citet{gonccalves2018exact} trade the expit link for a probit type of transformation $\sigma(\cdot) = \Phi(\cdot)$ where $\Phi$ is the standard normal CDF.

\subsection{Forward simulation}

The usual difficulty with Gaussian Cox processes is that the PPP likelihood depends on the whole of $g$ through the integral $ \int_\cS \lambda \sigma(g(s))\d s$ of the intensity function. It can be challenging to either conduct inference about the intensity function or integrate out $g$ without relying on some finite-dimensional approximation.  Yet, realizations of the SGCP, as defined above, can be simulated using a representation of the GP of random but finite dimension.

\begin{algo}[SGCP Simulation] The following procedure generates a realization of the SGCP:\label{SGCPalgo}
\begin{enumerate}
    \item Simulate a homogeneous PPP $X$ with intensity $\lambda$,
    \item Instantiate the GP $g(\cdot)$ at the locations of $X$ according to its finite dimensional Gaussian distribution,
    \item Keep every point $x\in X$ with probability $\sigma(g(x))$, otherwise discard $x$.
\end{enumerate}
\end{algo}

The intuition behind the above algorithm is as follows. Suppose you could simulate and access the value of $g(\cdot)$ at every location of $\cS$. In any case, the intensity function $\lambda \sigma (g(\cdot))$ is bounded by $\lambda$ and so step 2 of Definition \ref{SGCPdef} can be carried out by thinning a homogeneous PPP $X$ of intensity $\lambda$ ($X$ does not depend on $g(\cdot)$). Once $X$ is fixed, the thinning procedure operates without regard for the values of $g(\cdot)$ at other locations, hence the rest of the GP values are superfluous. In that sense, the procedure outlined above is a \textit{retrospective} sampler \citep{beskos2006retrospective,papaspiliopoulos2008retrospective}; the simulation of $g(\cdot)$ is delayed until it is only needed in a finite-dimensional form to perform the thinning of $X$. In Appendix \ref{valAlgo}, we demonstrate the validity of this procedure using only properties of Gaussian processes and point process distribution theory: the locations left at the end of Algorithm \ref{SGCPalgo} are indeed exactly distributed according to the Gaussian Cox process of Definition \ref{SGCPdef}.

Now consider $\tilde X_0, \tilde X_1$ as those point processes with GP distributed marks arising from the process $g(\cdot)$, where $\tilde X_0$ corresponds to locations that were thinned and $\tilde X_1$ is its observed counterpart (those retained at the end of Algorithm \ref{SGCPalgo}). Both point processes consist of a finite number of values in $\cS\times \R$. In particular, their distributions are not absolutely continuous wrt to the unit rate PPP distribution. Regardless, the joint density (in the sense described in Sections \ref{csDens} and \ref{Marks}) of $\tilde X_0, \tilde X_1$ exists and is immediately derived by using Theorem \ref{colThm}. Indeed, following the steps of Algorithm \ref{SGCPalgo}, we can generate the $\{0,1\}$ coloured point process by first drawing a Poisson distributed number of points that are subsequently scattered uniformly on $\cS$. The GP values can then be instantiated at those locations and finally each pair is coloured as observed with probability $\sigma(g(\cdot))$. Let $(x,\mathrm g, c)$ represent triplets in $\cS \times \R \times \{0,1\}$, the density of the marked point process described above is
\begin{align}
  f(\{(x_1,\mathrm g_1, c_1),\dots,(x_n,\mathrm g_n, c_n)\}) &= p_n \pi_n(x_1,\dots,x_n)\pi_n(\mathrm g_1,\dots,\mathrm g_n|x_1,\dots,x_n)\nonumber\\&\qquad\qquad\qquad\qquad\qquad \pi_n(c_1,\dots,c_n|x_1,\mathrm g_1,\dots,x_n,\mathrm g_n)\nonumber \\
  &= \frac{\exp(-\lambda|\cS|)\lambda^n}{n!} \cN(\mathrm g_1,\dots,\mathrm g_n|0,\Sigma(x_1,\dots,x_n)) \nonumber\\
  &\qquad\qquad\qquad\qquad\qquad\prod_{i=1}^n [\{1-\sigma(\mathrm g_i)\}^{1-c_i}\sigma(\mathrm g_i)^{c_i}]. \label{unitypeDensity}
\end{align}
Notice that the dependence of the GP values on the locations acts through the covariance matrix $\Sigma(x_1,\dots,x_n)$ which is computed from the covariance function $C(\cdot,\cdot)$ in the usual fashion. Also, recall that scattering models need to be equivalent under permutations of indices $\{1,2,\dots,n\}$ as is the case here.

\textbf{Note:} Expression \eqref{unitypeDensity} is the density used in \citet{gonccalves2018exact} (they also use a counting-scattering type of dominating measure). However, this expression cannot be used to describe conditional distributions for thinned and observed locations because it is not a proper joint density (wrt to a product measure); it is the density of a single, unitype point process with a $\{0,1\}$ valued mark. This is a subtle but important distinction.

Theorem \ref{colThm} tells us that the joint density of the thinned and observed point processes $\tilde X_0$ and $\tilde X_1$ is essentially the density of the $\{0,1\}$ marked point process multiplied by some combinatorial factor:
\begin{align}
f(\tilde S_0,\tilde S_1) &= \binom{n_0+n_1}{n_0,n_1}f(\{\tilde S_0\times \{0\}\}\cup \{\tilde S_1\times \{1\}\})\nonumber\\
&= \frac{\exp(-\lambda|\cS|)\lambda^{n_0+n_1}}{n_0!n_1!} \cN(\tilde {\mathrm g}_0,\tilde {\mathrm g}_1|0,\Sigma(\tilde x_0,\tilde x_1))
\prod_{i=1}^{n_0} \{1-\sigma(\mathrm g_{0,i})\}\prod_{j=1}^{n_1}\sigma(\mathrm g_{1,j}). \label{jointSGCP}
\end{align}
In the final expression, $\tilde S_k = \{(x_{k,1},\mathrm g_{k,1}),\dots,(x_{k,n_k},\mathrm g_{k,n_k})\}$ is a finite subset of $\cS\times\R$ while $\tilde x_k = (x_{k,1},\dots,x_{k,n_k})$ and $\tilde {\mathrm g}_k = (\mathrm g_{k,1},\dots,\mathrm g_{k,n_k})$ are respectively the location and mark components ($k=0,1$). Remember that the original objective was to obtain a tractable model for the augmented data $\tilde X_0, \tilde X_1$, and this is achieved: their joint density does not involve the integral of the intensity function. The key element of the colouring theorem in Section \ref{Col} is that deriving the joint density of the thinned and observed points becomes automatic.

Once the joint density of the thinned ($\tilde X_0$) and observed ($\tilde X_1$) locations has been established correctly, there is nothing controversial in using proportionality arguments to derive, up to a normalizing constant, the form of the conditional density of $\tilde X_0| \tilde X_1$. This is standard practice for deriving posterior distributions in Bayesian approaches.  By removing every factor that does not involve $\tilde S_0$ in expression \eqref{jointSGCP}, we can write the density of the thinned locations (with associated GP marks) conditional on their observed counterpart:
\begin{align}
f(\tilde S_0|\tilde S_1) &\propto \frac{\lambda^{n_0}}{n_0!} \cN(\tilde {\mathrm g}_0,\tilde {\mathrm g}_1|0,\Sigma(\tilde x_0,\tilde x_1))
\prod_{i=1}^{n_0} \{1-\sigma(\mathrm g_{0,i})\}.\label{condDens}
\end{align}
The expression above is perfectly tractable. It could be employed in Metropolis-Hastings methods specifically designed for point processes such as the birth-death-move algorithm of \citet{geyer1994simulation}.

In expression \eqref{condDens}, the GP values are intrinsically assumed to be known at the observed locations, that is $\tilde {\mathrm g}_1$ needs to be given to use this conditional density. Those do not correspond to any observable quantities and also need to be simulated. In practice \citep{adams2009tractable,rao2017bayesian,gonccalves2018exact}, this is handled by alternating simulations of the thinned point process $\tilde X_0|\tilde X_1$ from \eqref{condDens} and then all the GP values $\tilde {\mathrm g}_0,\tilde {\mathrm g}_1$ conditional on all the locations $\tilde x_0,\tilde x_1$, both thinned and observed. The thinned values $\tilde {\mathrm g}_0$ of the GP are effectively sampled twice at each single step of the Markov chain.

Simulating the GP values conditional on locations is not a dimension-changing move and does not present any particular issue beyond standard applications of Metropolis-Hastings. This stage boils down to a latent GP model and one can choose the most suitable method to sample $\tilde {\mathrm g}_0,\tilde {\mathrm g}_1$ (along with covariance parameters) from an abundant literature. Examples include the elliptical slice sampler \citep{murray2010slice} and Hamiltonian Monte Carlo \citep{hensman2015mcmc}. See \citet{filippone2013comparative} for an extensive study of such algorithms.

The conditional in \eqref{condDens} is the Gaussian density evaluated at the GP values $\tilde {\mathrm g}_0,\tilde {\mathrm g}_1$. This necessitates $\cO(n^3)$ operations to compute, where $n= n_0+n_1$ is the total number of observations (thinned and observed). Thus, the GP has to be instantiated at even more points than just the ones contained in the data and this number moves at every iteration. Moreover, if the observed points are configured in dense clusters, one might need a substantially larger number of thinned locations to fill in and make the base, uncoloured process homogeneous. When done repeatedly in an MCMC algorithm, the computational burden can accumulate quickly even if $n_1$ is only moderately large. \citet{shirota2019scalable} overcome this difficulty in the context of the SGCP using the nearest neighbor approximation \citep{datta2016hierarchical} to factorize the Gaussian density.

In the next section, we focus on the point process on point process conditional \eqref{condDens} and assume everything else to be fixed. We discuss how to simulate the thinned point process from this density. Multiple methods have been proposed in the literature, some of which are inconsistent with one another, and we aim to provide a unifying clarification.  See Appendix \ref{SampThinApp} for a full discussion.

\subsection{Sampling the Thinned Point Process}

\citet{adams2009tractable} employ birth-death-move proposals to simulate the thinned point process ($\tilde X_0$) conditional on the observed point process ($\tilde X_1$). The algorithm of \citet{adams2009tractable} can be interpreted as an application of existing trans-dimensional MCMC methodology. The acceptance ratios are identical to what would be obtained by applying the birth-death-move algorithm described in \citet{geyer1994simulation}. They are also the same ratios as one would get upon using the more general reversible jump MCMC methodology \citep{green1995reversible} (see example 3 in \citet{tierney1998note}).

Inserting, deleting and displacing one point at a time can add up to a challenging computational task. However, previous computations can be stored efficiently to curb the complexity of evaluating the Gaussian conditional. Since only one location is changed at a time in the birth-death-move algorithm, the inverse of the covariance matrix can be updated accordingly from the previous iteration using less burdensome computations. There is also the challenge of determining after how many birth-death-move proposals can we reasonably consider the resulting point process to be a new realization of the thinned locations. If there are too few, then consecutive iterations are composed of essentially the same points. 

The approach of \citet{adams2009tractable} is valid beyond any doubt in the sense that each update is reversible wrt the density in \eqref{condDens}. It would nevertheless be preferable to simulate directly, if possible, from the distribution of the thinned locations $\tilde X_1$ conditional on the observed points. \citet{gonccalves2018exact} attempt direct sampling by first simulating the number of locations and then simulating their position in $\cS$.  In their original publication, \citet{gonccalves2018exact} used a proportionality argument that omitted a term in the PMF of $n_0$. This step has been modified in \citet{gonccalves2023corrigendum}, which we discuss below.

\subsection{Retrospective sampling} \label{retSamp}

Conditional on the whole random field $g(\cdot)$, the thinned locations in $\cS$ are PPP($\lambda (1-\sigma \circ g(\cdot))$) distributed and are independent of the observed locations. This is implied by the original Poisson colouring theorem \citep{lewis1979simulation,kingman1992poisson}. In that case, the thinned locations can be simulated by thinning a Poisson process of homogeneous intensity $\lambda$. Consequently, only a finite representation of $g(\cdot)$ is required just as was the case in the retrospective sampler described by Algorithm \ref{SGCPalgo}.

This suggests a strategy to simulate from the conditional $\tilde X_0|\tilde X_1$. First, conceptually instantiate the GP from $g(\cdot)|\tilde X_1$. Second, sample from $\tilde X_0|g(\cdot)$ using Poisson thinning by uncovering the necessary values of $g()$ in the process. The difficulty, as we will see, lies in determining the finite-dimensional distributions of the random field $g(\cdot)|\tilde X_1$. On the one hand, $\tilde X_1$ consists of locations in $\cS$ with associated GP marks. Intuitively speaking, we could conclude that $g(\cdot)|\tilde X_1$ has the distribution of a GP conditional on some values being fixed. This is the argument invoked in \citet[Section 7.1]{rao2017bayesian}. It implies the following retrospective sampling procedure to update the thinned point process:
\begin{enumerate}
    \item Simulate a homogeneous PPP $X$ with intensity $\lambda$,
    \item Instantiate the GP $g(\cdot)$ at the locations of $X$ according to its finite-dimensional Gaussian distribution conditional on the observed values $\tilde {\mathrm g}_1 $ at the observed locations $\tilde x_1$.
    \item Keep each point of $\{(x,g(x)): x\in X\}$ with probability $1-\sigma(g(x))$.
\end{enumerate}
The procedure above defines an SGCP-like Cox process with stochastic intensity $\lambda\{1-\sigma(g_*(\cdot))\}$, where $g_*$ has the GP distribution of $g$ conditional on its value at the observed locations. Ultimately, the algorithm above is misconceived. Where the logic breaks down is that $g(\cdot)|\tilde X_1$ is not simply a GP. There is more information in $\tilde X_1$ than some values of $g(\cdot)$ at some locations of $\cS$. Those are values left after an independent thinning procedure with probabilities of being observed $ \sigma \circ g(\cdot)$. In some sense, we would already expect the values of the GP marks in $\tilde X_1$ to be high.

Another way to look at it is to consider the special case where the observed point process is empty. This situation is nonsensical from an applied perspective, but should nevertheless be handled by the general theory as the empty set is in the support of the observed point process $\tilde X_1$. According to the procedure in \citet{rao2017bayesian}, this would mean that the conditional $\tilde X_0| \tilde X_1 = \emptyset$ is the same as the marginal distribution of the thinned point process under the SGCP since there are no values upon which to condition the GP. This is somewhat suspicious as one would expect that 0 observations would be indicative of globally smaller values for the GP than under the marginal and, consequently, result in the conditional PMF for the number of thinned locations to be skewed towards higher values. In Appendix \ref{concMist}, we use this special case as a counter-example to the overall claim.

The approach of \citet{gonccalves2023corrigendum} is the same as the one in \citet{rao2017bayesian} with the exception that the GP values are simulated conditional on both the thinned and observed point processes.
\begin{algo}Update procedure for the thinned point process $\tilde X_0$ in \citet{gonccalves2023corrigendum}:\label{retroSampler}
   \begin{enumerate}
    \item Simulate a homogeneous PPP $X$ with intensity $\lambda$,
    \item Instantiate the GP $g(\cdot)$ at the locations of $X$ according to its finite-dimensional Gaussian distribution conditional on the thinned and observed values $\tilde {\mathrm g}_0,\tilde {\mathrm g}_1 $ at the observed locations $\tilde x_0,\tilde x_1$.
    \item Keep each point of $\{(x,g(x)): x\in X\}$ with probability $1-\sigma(g(x))$.
\end{enumerate}
\end{algo}

This procedure is similar to the birth procedure in that new points are simulated uniformly on $\cS$ and are assigned a mark by conditioning the Gaussian distribution on all other observations. It is not an exact sample from the full conditional \eqref{condDens} since the current state of $\tilde X_0$ plays a role in the update. The authors justify this method as a Gibbs sampler with the whole GP $g(\cdot)$ as one of the coordinates. The first step is to sample from $g(\cdot)|\tilde X_0,\tilde X_1$ which is performed retrospectively when sampling the Poisson process implied by $\tilde X_0|g(\cdot),\tilde X_1$. We would need to prove that the finite-dimensional distribution of $g(\cdot)$ conditional on the thinned and observed point process is indeed Gaussian as claimed. Recall that this was not the case when conditioning only on the observed points.

\textbf{Note:} In Appendix \ref{valRetroSampler}, we offer a more direct argument that foregoes retrospective sampling entirely and instead operates at the level of the two finite point processes $\tilde X_0$ and $\tilde X_1$. We show that Algorithm \ref{retroSampler} represents an update that is reversible wrt the conditional density in \eqref{condDens} and therefore leaves the target distribution invariant.

Rather than inserting or deleting one location at a time, the approach of \citet{gonccalves2023corrigendum} updates the thinned locations with an entirely new point process at each iteration. It is also relatively straightforward to implement which makes it an enticing option when employed, as intended, as a data augmentation step in an inference algorithm. We apply the spirit of their algorithm in the extension we present in the next section.

\section{Multitype Sigmoidal Gaussian Cox Process} \label{MTSGCP}

In this section, we propose an extension to the SGCP that handles multiple types of points. The Lansing Woods data from Section \ref{Intro} is publicly available from the \texttt{spatstat} package in R.  We conduct MCMC-based Bayesian inference on the multitype point process consisting of the position of maples and hickories.  This data set has been analyzed with a multitype log Gaussian Cox process (LGCP) model \citep{moller1998log,brix2001space} in \citet{waagepetersen2016analysis} using a minimum contrast estimator. The sampling distribution of such an estimator can be complicated. The authors use the parametric bootstrap to quantify uncertainty in the estimation of correlation parameters. Our Bayesian approach offers an alternative way to measure uncertainty in such models. One advantage of multitype Gaussian Cox processes such as the LGCP or our multitype version of the SGCP is that inter-process dependence can be relegated to the cross-covariance function of the latent multivariate Gaussian random field. There is a vast literature on multivariate GPs, see \citet{genton2015cross} for a review of some of the more common modeling approaches.

\subsection{Model Specification}

We generalize the SGCP model with probit link that was used in \citet{gonccalves2018exact} to model multitype point patterns. We assume at the first level that the point processes $X_1,\dots,X_K$ corresponding to the various types are independent and Poisson distributed conditional on some random and multivariate intensity function. Starting from a $K$-variate GP $g(\cdot)$, the function is defined by introducing a $K$-dimensional standard normal variable $Z$:
\begin{align}
  &X_j|g(\cdot) \indep \PPP(\lambda \sigma_j\circ g(\cdot)), \nonumber \\&\quad \text{where } \sigma_j(\mathrm g) = P\left(\left\{\argmax (\mathrm g + Z) = j \right\} \cap \left\{\max (\mathrm g + Z) > 0\right\}\right),   \label{mtSGCPppp}
\end{align}
for $j=1,2,\dots,K$. The $\sigma(\cdot)$ function maps a vector in $\R^K$ to a point on the $K$-dimensional probability simplex.

The $\sigma_j(\cdot)$ functions are defined by integrals of Gaussian densities over conic domains. They are cumbersome to compute whenever $K>1$. For this reason, we adapt the auxiliary variable formulation of the probit multinomial model \citep{mcculloch2000bayesian} to the spatial context by defining $Y(\cdot) = g(\cdot) + Z(\cdot)$. Here, $Z(\cdot)$ consists of $K$ independent white noise processes with unit variance. For $K=1$, we recover the data augmentation scheme of \citet{albert1993bayesian}. We can redefine the model in terms of this new random field $Y(\cdot)$:
\begin{align}
  X_j|Y(\cdot) \indep \PPP(\lambda \tau_j\circ Y(\cdot)), \quad \text{where } \tau_j(y) = \ind\left(\left\{\argmax (y) = j \right\} \cap \left\{\max (y) > 0\right\}\right),   \label{mtSGCPpppDA}
\end{align}
where the functions $\tau_j(\cdot)$ are binary valued and at most one entry from $j=1,2,\dots,K$ is non-null.

Adding a multivariate Gaussian random field distribution to $g(\cdot)$ completes the Cox process specification. We use the non-separable version of the linear model of coregionalization (LMC); see \citet{schmidt2003bayesian,gelfand2004nonstationary,gelfand2005spatial} for notable examples of its use in a Bayesian context. Such a process is constructed at each location from a linear transformation of independent random fields:
$$
g(s) = \begin{bmatrix}
g_1(s)\\
\vdots \\
g_K(s)\\
\end{bmatrix} = A \begin{bmatrix}
w_1(s)\\
\vdots \\
w_K(s)\\
\end{bmatrix} + \begin{bmatrix}
\mu_1\\
\vdots \\
\mu_K\\
\end{bmatrix}, \quad
 \text{with } w_j(\cdot) \indep \text{GP}(0,C_j(|\cdot|)).$$
The independent GPs $w_j(\cdot)$ above are parametrized by some stationary and isotropic correlation functions $C_j(|\cdot|)$ while scaling is handled by the full rank matrix $A$. In our implementation, we use exponential correlation functions with distinct range parameter $\rho_j > 0$, i.e. $C_j(r) = \exp(-\rho_j r), r \geq0$. We add type-specific mean levels $\mu_j,j=1,\dots,K$. We chose the LMC as a multivariate spatial model over alternatives because its covariance structure can be exploited computationally to obtain likelihood evaluations that are linear in $K$ rather than cubic \citep{alie2024computational}.

Same as in the univariate case described in Section \ref{SGCP}, the product of PPP likelihoods implied by \eqref{mtSGCPpppDA} involves the integral of the random intensity functions $\lambda\tau_j \circ Y(\cdot)$. On the other hand, we can interpret the model as arising from the colouring of a homogeneous PPP$(\lambda)$ process where a point at location $s\in \cS$ will be of species $j$ with probability $\sigma_j\circ g(\cdot)$ (equivalently if $\tau_j\circ Y(\cdot)=1$) independently of other points. This point will be of none of the $K$ species if all the components of $Y(s)$ are negative. We consider the auxiliary PPP $X_0$ of intensity $\lambda\tau_0 \circ Y(\cdot)$ with $\tau_0(y) =\ind\left(\max (y) < 0\right)$. This point process consists of fictional thinned locations. Heuristically, $X_0$ fills in the gaps to ensure that the union $\cup_{j=0}^K X_j$ is homogeneous.

We can devise a retrospective sampling procedure analogous to Algorithm \ref{SGCPalgo} to generate the point processes $\tilde X_0,\tilde X_1,\dots,\tilde X_K$ (including the thinned locations) along with the $K$-variate GP values at each location. To summarize, we can simulate $X\sim \PPP(\lambda)$ and instantiate $g(\cdot)$ at this finite number of locations in accordance with the multivariate normal distribution implied by the LMC. The values of the process $Y(\cdot)$ are obtained by adding to $g(\cdot)$ a standard $K$-variate normal random variable at each location in $X$. Each point $x \in X$ can then be classified in one of the $K+1$ categories along with its GP marks $g(x)$ and $Y(x)$ according to the rule defined by the $\tau_j(\cdot)$ functions.

Such a procedure operates without regard for the values of the GPs $g(\cdot)$ and $Y(\cdot)$ at locations of $\cS$ not included in $X$. Define $\tilde X_j$ as the points in $x \in X$ of the categorical mark (colour) $j$ along with their 2 corresponding multivariate GP values $g(x)$ and $Y(x)$ (each point has $2K$ marks). The density of the point processes $\tilde X_0, \tilde X_1, \dots, \tilde X_K$ wrt to the product of counting-scattering measures $\mu_{\text f}^{(K+1)}$ is obtained from the Colouring Theorem of Section \ref{Col}:
\begin{align}
f(\tilde S_0,\tilde S_1,\dots,\tilde S_K) = &\frac{\exp(-\lambda |\cS|) \lambda^{\sum_{j=0}^K n_j}}{\prod_{j=0}^K n_j !} \cN(\tilde{\mathrm g}_0,\tilde{\mathrm g}_1,\dots,\tilde{\mathrm g}_K|\tilde x_0,\tilde x_1,,\dots,\tilde x_K,\mu_{j=1}^K,A,\rho_{j=1}^K)\nonumber \\
&\quad \prod_{i=1}^{n_0} \cN(y_{0,i}|\mathrm g_{0,i},\mat I) \ind\left(\max (y_{0,i}) < 0\right)\nonumber \\
&\quad\prod_{j=1}^K \prod_{i=1}^{n_j} \cN(y_{j,i}|\mathrm g_{j,i},\mat I) \ind\left(\left\{\argmax (y_{j,i}) = j \right\} \cap \left\{\max (y_{j,i}) > 0\right\}\right). \label{completeDensity}
\end{align}
In this last expression, $\tilde S_j$ are finite subsets of $\cS \times \R^{2p}$ composed of physical locations $\tilde x_j = (x_{j,1},\dots,x_{j,n_j})$ along with marks $\mathrm{\tilde g}_j = (\mathrm g_{j,1},\dots,\mathrm g_{j,n_j})$ and $\tilde y_j = (y_{j,1},\dots,y_{j,n_j})$ for $j=1,\dots,K$. Importantly, including the point process $\tilde X_0$ makes the complete data likelihood free of intractable integral terms.

\subsection{Bayesian Inference} \label{bayInf}

A fully Bayesian specification is completed by assigning a prior to the model parameters $\lambda, \mu_{j=1}^K, A, \rho_{j=1}^K$. We respectively assign gamma and normal priors to $\lambda$ and $\mu_{j=1}^K$. This is mainly to ensure conjugate full conditional updates for these quantities. The components of $A$ are each assigned a normal prior. Finally, we put a uniform prior over $[3,30]$ on each of $\rho_{j=1}^K$ which roughly corresponds to a practical range (the distance at which correlation is equal to 0.05) between 0.1 and 1. We obtain a posterior sample by alternating updates that are reversible wrt the full conditionals of quantities of interest: the thinned point process $\tilde X_0$, the values $\tilde {\mathrm g}_0,\tilde {\mathrm g}_1,\dots,\tilde {\mathrm g}_K$ and $\tilde y_0,\tilde y_1,\dots,\tilde y_K$ of the latent GPs and the model parameters $\lambda, \mu_{j=1}^K, A, \rho_{j=1}^K$.

First, we perform the data augmentation step where the thinned point process $\tilde X_0$ is updated along with its GP marks. This is done, in the spirit of Algorithm \ref{retroSampler}, by first sampling a homogeneous PPP $X^*$ over $\cS$ of intensity $\lambda$. We then instantiate the random fields $g(\cdot)$ and $Y(\cdot)$ conditional on the GP observations contained in the current state of $\tilde X_0, \tilde X_1, \dots, \tilde X_K$. The new locations are those $x \in X^*$ for which $\max(Y(x))<0$. Those form the updated point process $\tilde X_0$ along with their $2K$ marks.

Next we sample all the GP values $\tilde {\mathrm g}_0,\tilde {\mathrm g}_1,\dots,\tilde {\mathrm g}_K$ and $\tilde y_0,\tilde y_1,\dots,\tilde y_K$ in a fixed dimension update. The marks of the thinned point processes are effectively sampled twice in a single update of the Markov chain. The values of the auxiliary GP $Y(\cdot)$ are sampled in turn from their univariate full conditional distribution which, from \eqref{completeDensity}, are truncated normal. Updating the values of $g(\cdot)$ boils down to instantiating random effects in a latent Gaussian model. We use the algorithm described in \citet{alie2024computational} rather than a straightforward multivariate normal update to avoid factorizing the $nK \times nK$ covariance matrix. Importantly, those potentially high dimensional quantities are all sampled from full conditionals. There are no proposal parameters that could otherwise be complicated to tune and tailor to any particular example.

As for the model parameters, we use regular Metropolis-Hastings to sample the range parameters $\rho_j$ in turn. The whole coregionalization matrix $A$ is updated using slice sampling \citep{neal2003slice,murray2010slice}. Finally, $\lambda$ and the mean vector $\mu$ are sampled from their conjugate full conditionals, respectively the gamma and multivariate Gaussian distributions. In the next section, we present our joint analysis of the two point configurations from Lansing Woods (showcased in Figure \ref{fig:mapHick}). To this end, we ran the algorithm outlined above for 10,000 iterations, discarding the first 2,000 as burn-in.

\subsection{Application}

The multitype extension allows us to fit the intensity function of the hickory and maple species simultaneously (see Figure \ref{fig:mapHickInt}). At each iteration of the Markov chain, the $K$-variate Gaussian process $g(\cdot)$ can be instantiated on a fine grid conditionally on their values $\tilde {\mathrm g}_0,\tilde {\mathrm g}_1, \dots,\tilde {\mathrm g}_K $ at thinned and observed locations. The transformation \eqref{mtSGCPppp} of $g(\cdot)$ that defines the intensity function is intractable. However, we can instantiate the $Y(\cdot)$ process on this grid by adding a $K$-dimensional standard normal at each location. We then compute the binary transformations $\tau_j(\cdot)$, multiply each one by the current value of $\lambda$ and ultimately average over all iterations of the Markov chain. Any desired resolution can be achieved although computations scale as the number of pixels cubed. This can technically be accomplished after generating the Markov chain using conditional properties of the multivariate normal distribution. However, it is more efficient to perform the calculation at each iteration so that the inverse covariance matrix at observed and thinned locations need not be recomputed or stored.

\begin{figure}[b!]
    \centering
    \includegraphics[width = 6cm]{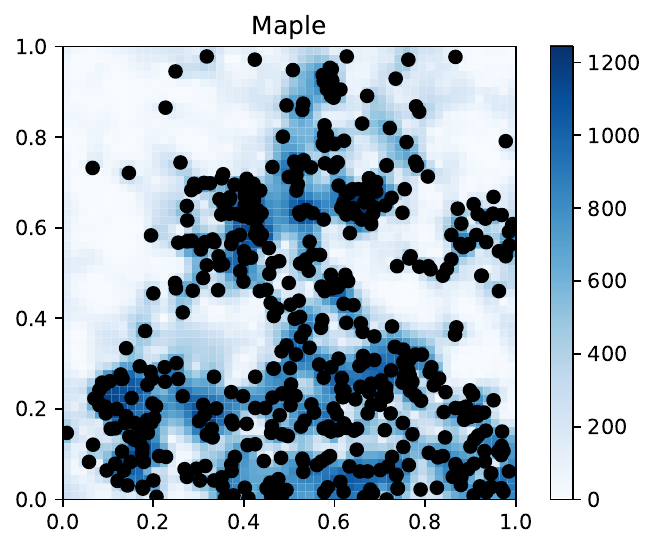}
    \includegraphics[width = 6cm]{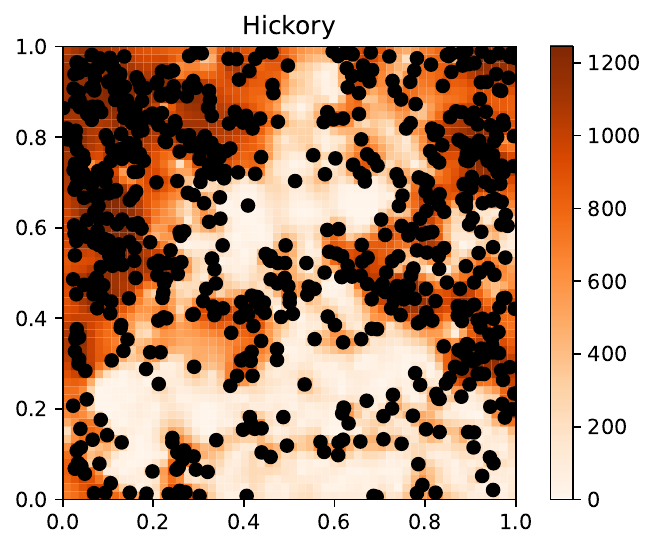}
    \caption{The pointwise posterior mean intensity function $\lambda \tau_j \circ Y(\cdot), j=1,2$ along with tree positions for maples (left) and hickories (right).}
    \label{fig:mapHickInt}
\end{figure}

Figure \ref{fig:mapHickInt} shows that the procedure is flexible enough to capture the particular growing pattern of the two types of trees through the functional form of their respective intensity. However, it tells us where maples and hickories grow, but it does not provide an answer as to why they grow in such a configuration. The univariate method of \citet{adams2009tractable} was devised to estimate the intensity function of a non-homogeneous PPP. In the multitype setting, the point processes are assumed independent at the first stage. In that sense, there is no added benefit in fitting both intensity functions jointly (unless we believed them to be dependent {\it a priori}).

Rather than considering the transformed GP as a flexible prior for the intensity function of Poisson distributed point patterns, we instead consider the model to be the Gaussian Cox process implied by integrating out the GP values. The parameters of interest are the base intensity $\lambda$ and those driving the LMC including the coregionalization matrix $A$, the range parameters $\rho_j$ and mean levels $\mu_j,j=1,\dots,K$. From the perspective of posterior sampling, nothing changes other than the fact that interpolating the intensity function on a fine grid might not be needed anymore. From this point of view, the particular pattern in which maple and hickories grow is irrelevant. It is not possible to discriminate between a Cox process and a non-homogeneous PPP from only one point configuration \citep{moller2003statistical,isham2010spatial}. The difference lies in the nature of repetition. If we were to look at a different patch of forest, maples and hickories would certainly grow into a different pattern. Nevertheless, we would expect to observe the same seemingly repulsive interaction between the two species.

We measure intra- and inter-species interaction with the (cross) pair correlation function (PCF):
\begin{align}
\gamma_{k\ell}(s,t) = \frac{E[\tau_k(Y(s)) \tau_\ell (Y(t))|A,\rho_{j=1}^K,\mu_{j=1}^K]}{E[ \tau_k(Y(s))|A,\rho_{j=1}^K,\mu_{j=1}^K]E[\tau_\ell (Y(t))|A,\rho_{j=1}^K,\mu_{j=1}^K]}, \qquad k,\ell=1,2,\dots,K.   \label{pcfnsi}
\end{align}
Roughly speaking, $\gamma_{k\ell}(s,t)$ is the probability of observing points of types $k,\ell$ at infinitesimal balls centered around $s$ and $t$ normalized by the product of marginal probabilities of observing them in such areas. This quantity is independent of the scale of the intensity function ($\lambda$ cancels out in this expression). The PCF as expressed in \eqref{pcfnsi} is a property of the Cox process model. It is a function of model parameters thereby expectations are taken conditionally on them being fixed. This caveat is only necessary in the Bayesian paradigm where parameters also have distributions.

Globally, the model is stationary and isotropic provided the covariance structure of the driving GPs is (so is the case with the exponential kernels $\exp(-\rho_j |\cdot|)$ we use in our LMC specification). In this case, $\gamma_{k \ell}(s,t)$ is a function of the distance $r = |s-t|$ between $s$ and $t$. The cross PCF \eqref{pcfnsi} does not admit a simple parametric form as a function of $r$. It can nevertheless be evaluated at any distance $r>0$ and for any set of parameters $A,\rho_{j=1}^K,\mu_{j=1}^K$ using a consistent Monte Carlo estimator. Reasonable precision is rapidly attained as we only rely on averages computed from repetitions of $2K$-dimensional Gaussian vectors. Those are distributed as two points of the $Y(\cdot)$ process separated by distance $r>0$. We evaluate the pointwise estimate of the PCF at multiple distances that we average over multiple sets of parameters sampled from their posterior distribution to account for Bayesian model uncertainty.

The estimated PCFs for the maple and hickory example are illustrated in Figure \ref{fig:pcfhickmap}. Stationary and isotropic Cox processes imply points of the same type that cluster at short range. We are more interested in the cross-pair correlation function $\gamma_{12}$; this quantity being smaller than 1 implies a repulsive interaction between maples and hickories. The dependence seems to diminish and become negligible at distances bigger than 0.4. An expanded discussion on the dependence structures that can be induced by our multitype SGCP is featured in Appendix \ref{depsSGCP}.

\begin{figure}[t!]
    \centering
    \includegraphics[height = 6cm]{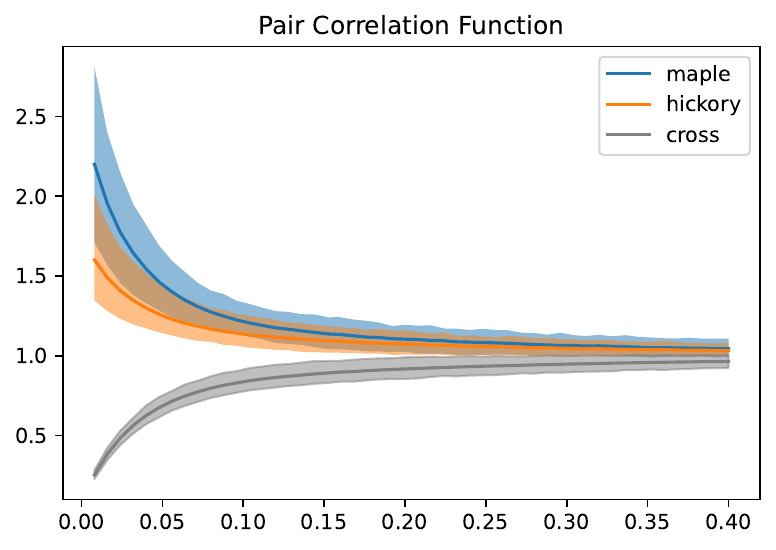}
    \caption{The pointwise posterior mean for the PCF $\gamma_{11}$ is illustrated in blue along with 90\% credible intervals computed at each point. Similarly, $\gamma_{22}$ is shown in orange and the cross pair correlation $\gamma_{12}$ is in gray.}
    \label{fig:pcfhickmap}
\end{figure}

\section{Discussion}

In this paper, we presented a general approach to obtain the distribution of a multitype point process defined through any discrete marking (colouring) of a base point configuration. This has implications for point process models defined by thinning procedures: the joint density of thinned and observed locations can be derived from the colouring theorem of Section \ref{Col}. Instantiating the thinned point process conditional on observed locations can help circumvent intractability in the marginal model and help carry inference of model parameters in a Bayesian framework.

Such data augmentation schemes have been employed in recent publications but in the absence of a unifying approach to the concept. For example, \citet{rao2017bayesian,gonccalves2018exact} introduce alternative methods to simulate the thinned locations in the context of the SGCP, but those turn out to be incompatible with the original formulation of \citet{adams2009tractable}. The alternative formulation introduced in the 2023 corrigendum of \citet{gonccalves2018exact} is perhaps the most enticing approach from a conceptual and computational standpoint. However, apart from the now retracted version of \citet{gonccalves2018exact}, none of the other proposals came with a formal proof of validity concerning the data augmentation update. The conceptual tools we presented in this paper allowed us to indubitably resolve these inconsistencies. Moreover, Theorem \ref{colThm} can be a very useful tool to design new multitype point process models with interesting dependence structures. In this regard, we introduce a generalization of the SGCP to jointly model multiple point patterns. We showcase our new model on the Lansing Woods data and conduct inference on the cross-pair correlation function of maples and hickories.

Besides correcting some misconceptions from the literature, our work can empower future authors to use more involved models and design even more efficient sampling algorithms without having to consider the intricate measure-theoretic details involved in defining joint, conditional and marginal densities for point processes. Interesting avenues to explore include thinning procedures applied to more complex models than Poisson point processes. For example, independently thinned Markov point processes have been introduced in the seminal work of \citet{baddeley2000non}. The general inference framework we provided might be an interesting alternative to the two-step semi-parametric procedure they present.


\subsubsection*{Funding}
Alie, Stephens and Schmidt acknowledge support from the Natural Sciences and Engineering Research Council of Canada (NSERC). Alie also acknowledges the support from the Fonds de recherche du Québec – Nature et technologies (FRQNT).


\bibliographystyle{plainnat}
\bibliography{references}

\appendix

\section{Proof of Theorem \ref{colThm}}\label{proColThm}

The following lemma encapsulates most of the technicalities involved in proving the colouring theorem of Section \ref{Col}. It tells us how to decompose an integral over the set $\cN_{\mathrm f}(\cS)$ of all finite subsets of $\cS$ into an integral over finite subsets of smaller spaces $\cS_0,\cS_1,\dots,\cS_K$.
\begin{lemma}\label{lem1}
    Let $\cS_0,\cS_1,\dots,\cS_K$ be any partition of the space $\cS$. The following holds for any positive function $h:\cN_{\mathrm f}(\cS) \tends [0,\infty)$:
    $$
    \int h(S)\mu_{\mathrm f}(\d S) = \int_{\cN_{\mathrm f}(\cS_K)} \dots\int_{\cN_{\mathrm f}(\cS_1)}\int_{\cN_{\mathrm f}(\cS_0)} \binom{\sum_{k=0}^K N_k}{N_0,N_1,\dots,N_K}  h\bigg(\bigcup_{k=0}^K S_k\bigg)\mu_{\mathrm f}(\d S_0)\mu_{\mathrm f}(\d S_1)\dots\mu_{\mathrm f}(\d S_K),
    $$
    where we use $N_k \equiv N(S_k)$ to denote the cardinality of the set $S_k$.
\end{lemma}

\begin{proof} It is sufficient to demonstrate the case for a binary partition $\cS_0,\cS_1$. The general case can then be obtained by further partitioning those sets. For a positive function $h$, the integral with respect to the counting-scattering measure can be computed as
\begin{align}
    \int h(S)\mu_{\mathrm f}(\d S) = \sum_{n\geq 0} \int_{\cS^n} h(\{x_1,x_2,\dots,x_n\}) \mu(\d x_1)\mu(\d x_2)\dots\mu(\d x_n) \label{CSint}
\end{align}
where $\mu$ is a measure over subsets of $\cS$. This is the extension of equation \eqref{CSmeas} from indicator to positive functions by the monotone class theorem.

For every $n\geq 0$, the inside integral on the RHS of \eqref{CSint} can be split into the $2^n$ parts defined by whether each of the coordinates belong to $\cS_0$ or $\cS_1$:
\begin{align}
\sum_{\mat k\in \{0,1\}^n}\int_{\cS_{k_n}} \dots \int_{\cS_{k_2}} \int_{\cS_{k_1}} h(\{x_1,x_2,\dots,x_n\}) \mu(\d x_1)\mu(\d x_2)\dots\mu(\d x_n). \label{partInt}
\end{align}
Recall that $h$ is a set function and the arbitrary ordering $1,2,\dots,n$ is irrelevant to its value. It is symmetric wrt to the integral coordinates $x_1,x_2,\dots,x_n$.

Therefore, to evaluate any of the summands in \eqref{partInt}, it matters only to know how many coordinates among the $n$ belong to $\cS_0$ and, accordingly, how many belong to its complement $\cS_1$. We can thus catalog each term in summation \eqref{partInt} by first noting how many $m\in\{0,1,\dots,n\}$ of its coordinates are integrated over $\cS_0$ and then sum the $\binom{n}{m}$ equivalent components:
\begin{align}
    \sum_{m=0}^n \binom{n}{m} \int_{\cS_0^m}\int_{\cS_1^{n-m}}h(\{x_1,x_2,\dots,x_n\}) \mu(\d x_1)\mu(\d x_2)\dots\mu(\d x_n). \label{symInt}
\end{align}

Substituting \eqref{symInt} for the integral on the RHS of \eqref{CSint}, we obtain
\begin{align*}
    \int h(S)\mu_{\mathrm f}(\d S) = \sum_{n\geq 0}  \sum_{m=0}^n \binom{n}{m} \int_{\cS_0^m}\int_{\cS_1^{n-m}}h(\{x_1,x_2,\dots,x_n\}) \mu(\d x_1)\mu(\d x_2)\dots\mu(\d x_n).
\end{align*}
We can re-index the double sum on the RHS using $n_0=m$ and $n_1=n-m$ and get
$$
\sum_{n_0\geq 0}  \sum_{n_1\geq 0} \binom{n_0+n_1}{n_0,n_1} \int_{\cS_0^{n_0}}\int_{\cS_1^{n_1}}h(\{x_1,x_2,\dots,x_{n_0+n_1}\}) \mu(\d x_1)\mu(\d x_2)\dots\mu(\d x_{n_0+n_1}).
$$
This last expression is equivalent to the double integral
$$
\int_{\cN_{\mathrm f}(\cS_1)}\int_{\cN_{\mathrm f}(\cS_0)} \binom{N_0 + N_1}{N_0,N_1}  h(S_0 \cup S_1)\mu_{\mathrm f}(\d S_0)\mu_{\mathrm f}(\d S_1),
$$
which completes the proof.

\end{proof}

We now prove the colouring theorem of Section \ref{Col}. Recall that this result relates the density of a point process $\tilde X$ over domain $\cS$ with mark space $\{0,1,\dots,K\}$ to the multitype process $X_0,X_1,\dots,X_K$ consisting of the locations in $\cS$ with those respective marks.

\begin{proof}[Proof of Theorem \ref{colThm}]
Take any sets $F_0,F_1,\dots,F_K\subseteq \cN_\text{f}(\cS)$, it is sufficient to verify that the distribution
\begin{align}
    P(X_0\in F_0, X_1 \in F_1,\dots,X_K\in F_K) \label{jointDist}
\end{align}
is equal to the integral
$$
\int_{F_K}\dots\int_{F_1}\int_{F_0} \binom{\sum_{k=0}^K N_k}{N_0,N_1,\dots,N_K} f\bigg(\bigcup_{k=0}^K S_k\times \{k\}\bigg)\mu_{\mathrm f}(\d S_0)\mu_{\mathrm f}(\d S_1)\dots\mu_{\mathrm f}(\d S_K),
$$
where $f$ is the density of the base process $\tilde X$ wrt to $\tilde\mu_f$: the counting-scattering measure constructed from the product of $\mu$ with the counting measure, i.e.
$$
\tilde\mu_\text{f}(\tilde F) = \sum_{n\geq 0} \int_{\cS^n}\sum_{\mat c \in \{0,1,\dots,K\}^n} \ind_{\tilde F}(\{(s_1,c_1),\dots,(s_n,c_n)\}) \mu(\d s_1) \dots \mu(\d s_n).
$$
The result can then be extended from rectangle sets of the form $F_0\times F_1\times\dots\times F_K$ to any set in the product $\sigma$-field of $X_0,X_1,\dots,X_K$ by the $\pi$-$\lambda$ Theorem \citep[Theorem A.1.4]{durrett2019probability}.

We first transform equation \eqref{jointDist} into a statement $P(\tilde X \in \tilde F)$ about the distribution of the base point process for some carefully chosen $\tilde F \subseteq \cN_\text{f}(\cS \times \{0,1,\dots,K\})$. For all $k=0,1,\dots,K$, we need to have $\tilde X \cap \{\cS\times \{k\}\} = S_k \times \{k\}$ for some finite set $S_k\in F_k$. Equivalently, it means that for
\begin{align}
    \tilde F := \bigg\{\bigcup_{k=0}^K S_k\times \{k\}:S_k \in F_k, k=0,1,\dots,K\bigg\}, \label{defTF}
\end{align}
we have $\tilde X \in \tilde F \iif X_0\in F_0, X_1 \in F_1,\dots,X_K\in F_k$.

By definition of $f$, we have
$$
P(\tilde X \in \tilde F) = \int \ind_{\tilde F}(\tilde S) f(\tilde S) \mu_\text{f}(\d \tilde S)
$$
which, for the partition $\tilde \cS_0 = \cS \times \{0\},\tilde \cS_1 =\cS \times \{1\},\dots,\tilde \cS_K =\cS \times \{K\}$, can be decomposed as
\begin{align}
    \int_{\cN_\text{f}(\tilde \cS_K)} \dots \int_{\cN_\text{f}(\tilde \cS_1)}\int_{\cN_\text{f}(\tilde \cS_0)}\binom{\sum_{k=0}^K N_k}{N_0,N_1,\dots,N_K}\ind_{\tilde F}\bigg(\bigcup_{k=0}^K \tilde S_k\bigg) f\bigg(\bigcup_{k=0}^K \tilde S_k\bigg) \tilde\mu_\text{f}(\d \tilde S_0)\tilde\mu_\text{f}(\d \tilde S_1)\dots \tilde\mu_\text{f}(\d \tilde S_K) \label{appLem}
\end{align}
according to Lemma \ref{lem1}. We can verify that an integral over the set of finite subsets of $\tilde\cS_k = \cS \times \{k\}$ for any fixed $k =0,1,\dots K$ can be expressed as an integral over $\cN_\text f(\cS)$ as
$$
\int_{\cN_\text{f}(\tilde\cS_k)} h(\tilde S_k) \tilde\mu_\text f(\d \tilde S_k) = \int h( S_k \times \{k\}) \mu_\text f(\d  S_k)
$$
for any integrable $h$, where $\mu_\text f$ is the counting-scattering measure over subsets of $\cS$. Applying this fact to equation \eqref{appLem}, we obtain
\begin{align}
P(\tilde X \in \tilde F) = \int\binom{\sum_{k=0}^K N_k}{N_0,N_1,\dots,N_K}\ind_{\tilde F}\bigg(\bigcup_{k=0}^K  S_k \times \{k\}\bigg) f\bigg(\bigcup_{k=0}^K  S_k \times \{k\}\bigg) \mu_\text{f}(\d  S_0)\mu_\text{f}(\d  S_1)\dots \mu_\text{f}(\d S_K).  \label{distRect}
\end{align}

Finally, per definition \eqref{defTF}, we have that $\bigcup_{k=0}^K  S_k \times \{k\}$ will be in $\tilde F$ iif $S_k\in F_k$ for all $k=0,1,\dots,K$. Substituting $\ind_{F_0\times F_1 \times \dots \times F_K}(S_0,S_1,\dots,S_K)$ for $\ind_{\tilde F}(\bigcup_{k=0}^K  S_k \times \{k\})$ in \eqref{distRect} completes the proof.

\end{proof}

\section{Validity of Algorithm \ref{SGCPalgo}}\label{valAlgo}

In this section, we demonstrate the equivalence between the marginal distribution of observed locations in $\cS$ implied by Algorithm \ref{SGCPalgo} and the Cox process of Definition \ref{SGCPdef}. The starting point is the density \eqref{jointSGCP} of $\tilde X_0,\tilde X_1$ which are respectively the thinned and observed point processes with values in $\cN_\text f(\cS\times \R)$. We will use $X_1$ to denote the observed locations without the Gaussian process marks.

We need to express the distribution $P(X_1 \in F)$ for some arbitrary $F\subseteq \cN_\text f(\cS)$ in terms of the joint distribution of $\tilde X_0,\tilde X_1$. The thinned locations can take any value but the observed locations along with GP marks need to be in a set $\tilde F$ comprised of any finite set $\{(x_1,g_1),\dots,(x_n,g_n)\}$ such that $\{x_1,\dots,x_n\}\in F$ and $(g_1,\dots,g_n)\in \R^n$. We are looking for
\begin{align}
P(X_1 \in F) &= P(\tilde X_0 \in \cN_\text f(\cS\times \R),\tilde X_1 \in \tilde F) \nonumber\\
    &=\sum_{n_1\geq0} \sum_{n_0\geq0} \int_{\cS^{n_1}}\int_{\cS^{n_0}} \ind_{ F}(\{(x_{1,1},\dots,x_{1,n_1}\})\frac{\exp(-\lambda|\cS|)\lambda^{n_0+n_1}}{n_0!n_1!}\nonumber\\
    &\qquad\qquad\qquad\qquad\qquad \int_{\R^{n_1}}\int_{\R^{n_0}} \cN(\tilde {\mathrm g}_0,\tilde {\mathrm g}_1|0,\Sigma(\tilde x_0,\tilde x_1))
\prod_{i=1}^{n_0} \{1-\sigma(\mathrm g_{0,i})\}\nonumber\\
&\qquad\qquad\qquad\qquad\qquad\qquad \prod_{j=1}^{n_1}\sigma(\mathrm g_{1,j}) \d\tilde {\mathrm g}_0 \d\tilde {\mathrm g}_1 \d \tilde x_0\d \tilde x_1. \label{grosseBertha}
\end{align}
Importantly, we used the fact that
$$
\ind_{\tilde F}(\{(x_{1,1},\mathrm g_{1,1}),\dots,(x_{1,n_1},\mathrm g_{1,n_1})\}) = \ind_{ F}(\{x_{1,1},\dots,x_{1,n_1}\})
$$
since, by definition of $\tilde F$, it only matters that locations $\{x_{1,1},\dots,x_{1,n_1}\}$ are in $F$. The GP marks can take any value. Using the defining property of Gaussian processes, we can replace the integral over $\R^{n_0+n_1}$ in \eqref{grosseBertha} by the expectation
$$
E\bigg[\prod_{i=1}^{n_0} \{1-\sigma( g(x_{0,i}))\}\prod_{j=1}^{n_1}\sigma(g(x_{1,j}))\bigg]
$$
with $g\sim\GP(0,C(\cdot,\cdot))$.

By Fubini’s theorem for positive functions, we can relegate this expectation to the outside and choose to first compute the integral over the thinned locations:
$$
\sum_{n_0\geq0} \int_{\cS^{n_0}} \frac{\lambda^{n_0}}{n_0!}\prod_{i=1}^{n_0} \{1-\sigma( g(x_{0,i}))\}\d \tilde x_0= \exp\Big(\lambda\int_\cS 1-\sigma(g(s))\d s\Big).
$$

Returning to equation \eqref{grosseBertha}, we are left with
$$
P(X_1 \in F) = E\bigg[ \sum_{n_1\geq 0} \int_{\cS^{n_1} } \ind_{ F}(\{(x_{1,1},\dots,x_{1,n_1}\})\frac{\exp(-\lambda\int \sigma(g(s))\d s)\lambda^{n_1}}{n_1!}\prod_{j=1}^{n_1}\sigma(g(x_{1,j}))\d \tilde x_1 \bigg].
$$
This is what we set out to prove as the above equation is simply the expectation over the GP function $g$ of the $ \PPP(\lambda \sigma (g(\cdot)))$ distribution. This is exactly the marginal implied by Definition \ref{SGCPdef}.

\section{Sampling the Thinned Process: Algorithms} \label{SampThinApp}

\textbf{\citet{adams2009tractable}:} The authors employ birth-death-move proposals to simulate the thinned point process ($\tilde X_0$) conditional on the observed point process ($\tilde X_1$). When a birth is proposed, a new location $x_\text{birth}$ in $\cS \times \R$ is simulated by first choosing a random point $x_\text{birth}$ uniformly in $\cS$. Its associated mark $\mathrm g_\text{birth}$ is then drawn from the Gaussian distribution conditioned on both thinned and observed marks $\tilde {\mathrm g}_0, \tilde {\mathrm g}_1$ at locations $\tilde x_0,\tilde x_1$. This move is accepted with probability
$$
\min \left(1, \frac{\lambda|\cS|}{n_0+1} \frac{1-b(n_0+1)}{b(n_0)} (1-\sigma(\mathrm g_\text{birth}))\right),
$$
where $b(n_0)$ is the probability that a birth is proposed if there are $n_0$ thinned locations at the current state of the Markov chain. When a death is proposed, a random point $(x_\text{death}, \mathrm g_\text{death})$ is chosen uniformly among the current thinned locations and its removal is accepted with probability
$$
\min \left(1, \frac{n_0}{\lambda|\cS|} \frac{b(n_0-1)}{1-b(n_0)} \frac{1}{1-\sigma(\mathrm g_\text{death})}\right).
$$
The locations can also be displaced. The acceptance rate in this case can be deduced from the regular, fixed dimension Metropolis-Hastings procedure.

\textbf{\citet{gonccalves2018exact}:} The authors attempt direct sampling by first simulating the number of locations and then simulating their position in $\cS$.  To do this, we need to derive both the PMF for the total number of locations and the scattering distributions for $\tilde X_0| \tilde X_1$: the point process with density
\begin{align*}
    f(\tilde S_0|\tilde S_1) = c \frac{\lambda^{n_0}}{n_0!} \cN(\tilde {\mathrm g}_0,\tilde {\mathrm g}_1|0,\Sigma(\tilde x_0,\tilde x_1))
\prod_{i=1}^{n_0} \{1-\sigma(\mathrm g_{0,i})\},
\end{align*}
where $c$ is a constant that does not depend on $\tilde S_0$. The scattering density over $(\cS \times \R)^{n_0}$ for a given $n_0$ is proportional to $\cN(\tilde {\mathrm g}_0,\tilde {\mathrm g}_1|0,\Sigma(\tilde x_0,\tilde x_1))
\prod_{i=1}^{n_0} \{1-\sigma(\mathrm g_{0,i})\}$ and can be simulated from by using a rejection sampling algorithm. However, the normalizing constant of this scattering density is a function of $n_0$. It means that the PMF for the total number $n_0$ of thinned locations is of the form
$$
p_{n_0} \propto \frac{\lambda^{n_0}}{n_0!} \int_{(\cS\times \R)^{n_0}} \cN(\tilde {\mathrm g}_0,\tilde {\mathrm g}_1|0,\Sigma(\tilde x_0,\tilde x_1))
\prod_{i=1}^{n_0} \{1-\sigma(\mathrm g_{0,i})\} \d \tilde x_0 \d \tilde {\mathrm g}_0,
$$
which is highly intractable.

\section{Conceptual Mistake in \citet[Section 7.1]{rao2017bayesian}}\label{concMist}

The thinning algorithm proposed by \citet[Section 7.1]{rao2017bayesian} (see Section \ref{retSamp}) does not simulate from the distribution (described by \eqref{condDens}) of thinned locations $\tilde X_0$ conditional on observed points $\tilde X_1$. To see this, we consider the simplest counterexample.

Consider the absurd but nonetheless valid situation where the observed point process is empty, i.e. $\tilde X_1 = \emptyset$. Under these circumstances, there are no observed GP values to condition upon in step 2 and therefore the marks are distributed in accordance with the random field $g \sim \GP(0,C(\cdot,\cdot))$. In that case, the conditional distribution of $\tilde X_0|\tilde X_1 = \emptyset$ would be the same as the marginal of $\tilde X_0$. We could ask: what is the probability of observing $\tilde X_0=\emptyset$ given that $\tilde X_1=\emptyset$? If the procedure outlined in \citet{rao2017bayesian} is valid, it would mean that
\begin{align}
f(\emptyset,\emptyset) = \int f(\emptyset,\tilde S_1) \mu_\text f(\d \tilde S_1) \int f(\tilde S_0,\emptyset) \mu_\text f(\d \tilde S_0). \label{falseClaim}
\end{align}
In this last expression, $f(\emptyset,\emptyset)$ is the joint density \eqref{jointSGCP} of $\tilde X_0, \tilde X_1$ evaluated at empty sets, and the integral terms are the marginal probabilities of being empty for the thinned and observed point processes under the SGCP. We demonstrate that equality \eqref{falseClaim} does not hold.

The RHS of equation \eqref{falseClaim} has the form
\begin{align*}
\sum_{n_0 \geq 0}\sum_{n_1 \geq 0}  \frac{\exp(-2\lambda |\cS|)\lambda^{n_0+n_1}}{n_0!n_1!} \int_{\cS^{n_0+n_1}} \int_{\R^{n_0+n_1}} &\cN(\tilde {\mathrm g}_0|0,\Sigma(\tilde x_0)) \cN(\tilde {\mathrm g}_1|0,\Sigma(\tilde x_1))\\
&\quad \prod_{i=1}^{n_0}\{1-\sigma(\mathrm g_{0,i})\} \prod_{j=1}^{n_1}\sigma(\mathrm g_{1,j}) \d \tilde {\mathrm g}_0 \d  \tilde {\mathrm g}_1 \d \tilde x_0  \d \tilde x_1.
\end{align*}
By the defining property of the GP $g$, this last expression is equal to
\begin{align*}
\sum_{n_0 \geq 0}\sum_{n_1 \geq 0}  \frac{\exp(-2\lambda |\cS|)\lambda^{n_0+n_1}}{n_0!n_1!} \int_{\cS^{n_0+n_1}}
E\bigg[\prod_{i=1}^{n_0}\{1-\sigma( g(x_{0,i}))\} \bigg] E\bigg[\prod_{j=1}^{n_1}\sigma(g(x_{1,j}))\bigg]   \d \tilde x_0  \d \tilde x_1.
\end{align*}

By Fubini's theorem for positive functions, we can pull the sum and integral over the thinned point process  inside the first expectation and obtain
\begin{align*}
    E\bigg[\sum_{n_0 \geq 0} \frac{\lambda^{n_0}}{n_0!}  \Big\{\int_\cS 1-\sigma( g(x)) \d x \Big\}^{n_0} \bigg] &=  E\bigg[\exp\Big(\lambda\int_\cS 1-\sigma( g(x)) \d x\Big)\bigg],
\end{align*}
where we used the series representation of the exponential function. Likewise, we can also put the integral and sum over the observed point process inside the second expectation:
\begin{align*}
    E\bigg[\sum_{n_1 \geq 0} \frac{\lambda^{n_1}}{n_1!}  \Big\{\int_\cS \sigma( g(x)) \d x \Big\}^{n_1} \bigg] &=  E\bigg[\exp\Big(\lambda\int_\cS \sigma( g(x)) \d x\Big)\bigg].
\end{align*}

Finally, we can bound the product of the two expectations below by employing Jensen's inequality on each one. By doing so, we obtain
$$
 \exp(-\lambda|\cS|) < \int f(\emptyset,\tilde S_1) \mu_\text f(\d \tilde S_1) \int f(\tilde S_0,\emptyset) \mu_\text f(\d \tilde S_0),
$$
where we used the linear property of expectations inside the exponential function. The LHS of the inequality above is exactly $f(\emptyset,\emptyset)$, and therefore this debunks equation \eqref{falseClaim}. This is in agreement with the intuition in Section \ref{retSamp}: the probability of the thinned locations being empty given that the observed point process is empty should be lower than under the marginal distribution of the thinned locations.

\section{Validity of Algorithm \ref{retroSampler}}\label{valRetroSampler}

Algorithm \ref{retroSampler} is a thinning algorithm that is almost identical to Algorithm \ref{SGCPalgo}: the original algorithm for exact simulations of the SGCP. The difference is that Gaussian process marks are instantiated conditional on the values of $\tilde {\mathrm g}_0,\tilde {\mathrm g}_1$ at respective locations $\tilde x_0,\tilde x_1$. This GP information corresponds to what is contained in the current state of the thinned ($\tilde X_0$) and observed ($\tilde X_1$) point processes.

The joint density (wrt to the product measure $\mu_\text f \times \mu_\text f$) of the new thinned and observed point processes resulting from Algorithm \ref{retroSampler} has a form that is very similar to \eqref{jointSGCP}:
\begin{align}
f'(\tilde S_0',\tilde S_1') = \frac{\exp(-\lambda|\cS|)\lambda^{n_0'+n_1'}}{n_0'!n_1'!} \frac{\cN(\tilde {\mathrm g}_0',\tilde {\mathrm g}_1',\tilde {\mathrm g}_0,\tilde {\mathrm g}_1|0,\Sigma(\tilde x_0',\tilde x_1',\tilde x_0,\tilde x_1))}{\cN(\tilde {\mathrm g}_0,\tilde {\mathrm g}_1|0,\Sigma(\tilde x_0,\tilde x_1))}
\prod_{i=1}^{n_0'} \{1-\sigma(\mathrm g_{0,i}')\}\prod_{j=1}^{n_1'}\sigma(\mathrm g_{1,j}'). \label{jointRetro}
\end{align}
where $\tilde S_k' = \{(x_{k,1}',\mathrm g_{k,1}'),\dots,(x_{k,n_k}',\mathrm g_{k,n_k}')\}$ is a finite subset of $\cS\times\R$ while $\tilde x_k' = (x_{k,1}',\dots,x_{k,n_k}')$ and $\tilde {\mathrm g}_k' = (\mathrm g_{k,1}',\dots,\mathrm g_{k,n_k}')$ are respectively the vectorized form of the location and mark components ($k=0,1$). Conditional on the current thinned point process, the update density is obtained by integrating out the observed coordinate $\tilde S_1'$ from \eqref{jointRetro}:
$$
q(\tilde S_0 \rightarrow \tilde S_0') = \int f'(\tilde S_0',\tilde S_1') \mu_\text f(\d \tilde S_1').
$$

From \eqref{condDens}, we have
\begin{align*}
f(\tilde S_0|\tilde S_1) q(\tilde S_0 \rightarrow \tilde S_0') \propto &\int\frac{\lambda^{n_0 + n_0'+ n_1'}}{n_0!n_0'!n_1'!} \cN(\tilde {\mathrm g}_0',\tilde {\mathrm g}_1',\tilde {\mathrm g}_0,\tilde {\mathrm g}_1|0,\Sigma(\tilde x_0',\tilde x_1',\tilde x_0,\tilde x_1)) \\
&\quad\prod_{i=1}^{n_0} \{1-\sigma(\mathrm g_{0,i})\} \prod_{i=1}^{n_0'} \{1-\sigma(\mathrm g_{0,i}')\}\prod_{j=1}^{n_1'}\sigma(\mathrm g_{1,j}')  \mu_\text f(\d \tilde S_1').
\end{align*}
The expression above is symmetric wrt $\tilde S_0$ and $\tilde S_0'$ and therefore the update is reversible wrt the conditional density \eqref{condDens}.

\section{Mat\'ern Type III Process}\label{MT3sect}

In their article, \citet{rao2017bayesian} showcase another use of data augmentation in the context of Bayesian inference for point processes. The model they discuss is the Mat\'ern type III process: one of the three repulsive processes described in \citet{matern1960spatial}. In their basic form, the Mat\'ern processes do not allow points to fall within a distance $R$ of one another which leads to regular or underdispersed realizations. Each one is described by a thinning procedure applied to a base PPP. We focus on the type III process as it has a tractable distribution for the thinned locations conditional on the observed points. This is one of the main results in \citet{rao2017bayesian}. We re-derive it here using the theory developed in Section \ref{ppsect} to illustrate how the colouring theorem streamlines most of the intricate details involved.

The thinning procedure involved in generating a Mat\'ern type III process proceeds as follows. First, we simulate a homogeneous PPP of intensity $\lambda>0$ on $\cS$. Then, each point is independently assigned a time mark from the uniform distribution on $[0,1]$. In the product space representation, this amounts to a homogeneous PPP on $\cS\times[0,1]$ with intensity $\lambda$. Finally, browsing through the points in order of time, each point is thinned if it lies within a radius $R$ of an earlier non-deleted point. The base point process is Poisson, but the locations are not independently thinned. We use the general version of the colouring theorem presented in Section \ref{Col}.

The fact that $\cS$ is bounded entails edge effects in the Mat\'ern type III process. Indeed, those points near the border have less space around them and are therefore overall less likely to be thinned. In practice, we observe points in an arbitrary bounded window contained in a larger, conceptually unbounded, space where the process still unfolds. There is no magical border that would prevent a point in $\cS$ from being thinned by a location outside of it. In the case of the Mat\'ern type III process, there is no easy fix to this phenomenon as the distance at which points can interact is unbounded. For example, a point near the border of $\cS$ could be thinned by an outside point within distance $R$ of it, but only if this point had not been thinned itself by another within distance $R$ and so on. Perfect simulation algorithms (without edge effects) of the Mat\'ern type III process and other details are discussed in \citet{moller2010perfect}. We focus on the simplistic case as it is set up in \citet{huber2009likelihood} and \citet{rao2017bayesian} to demonstrate the results therein as particular applications of the colouring theorem.

We proceed exactly as in Section \ref{SGCP} for the SGCP. This approach applies in general to any thinning procedure over any FPP with a known density including the other two types of Mat\'ern processes. We start by deriving the density of the base, $\{0,1\}$ marked point process from the generative procedure. From it, we obtain the joint density for the thinned and observed point processes by applying Theorem \ref{colThm}. Once we have the proper joint density, we can discuss conditional and marginal densities without any ambiguity by using standard definitions. The generative procedure of the discretely marked point process corresponding to Mat\'ern type III thinning can be condensed into the following two steps.
\begin{enumerate}
    \item Simulate a homogeneous PPP $\tilde X$ with intensity $\lambda$ over $\cS \times [0,1]$,
    \item Cycling through all the points $(s,t)\in \tilde X$ in order of time, assign label $c=0$ (thinned) to $(s,t)$ if it lies within distance $R$ of an earlier point that was not thinned. Otherwise, $c=1$ (observed).
\end{enumerate}

The density of the (unitype) FPP over $\cS \times [0,1] \times \{0,1\}$ implied by the above procedure is the product of a homogenous PPP density with the proper scattering distribution $\pi_n(c_1,\dots,c_n|s_1,t_1,\dots,s_n,t_n)$ for some $n\geq 1$ with $(c_1,\dots,c_n) \in \{0,1\}^n$. This latter function accounts for how points in $\cS \times [0,1]$ are assigned the label 0 (thinned) or 1 (observed). Those are allocated sequentially in order of time in a deterministic fashion. Consider the permutation $\sigma$ in which times are sorted in strictly ascending order, that is $t_{\sigma(1)} < t_{\sigma(2)} < \dots < t_{\sigma(n)}$. Browsing through the points in $\sigma$ order, the $i^{\text{th}}$ point is kept ($c_{\sigma(i)}=1$) if
$$
\prod_{j=1}^{i-1} \{1-\ind(||s_{\sigma(i)}-s_{\sigma(j)}||< R)\}^{c_{\sigma(j)}} =1,
$$
where the exponent $c_{\sigma(j)}$ accounts for the fact that a point can only be thinned by previous observed points (we use $0^0=1$ as a convention). The label scattering distribution can thus be factored in this $\sigma$ order as
\begin{align}
\prod_{i=1}^n \bigg(&\bigg[1-\prod_{j=1}^{i-1} \{1-\ind(||s_{\sigma(i)}-s_{\sigma(j)}||< R)\}^{c_{\sigma(j)}}\bigg]^{1-c_{\sigma(i)}}\nonumber\\ &\qquad\qquad\bigg[\prod_{j=1}^{i-1} \{1-\ind(||s_{\sigma(i)}-s_{\sigma(j)}||< R)\}^{c_{\sigma(j)}}\bigg]^{c_{\sigma(i)}}\bigg). \label{labScatOrd}
\end{align}
The value of this density whenever such an order $\sigma$ cannot be defined (in the case of ties) is not important because the dominating measure allocates zero mass to such cases.

Defining the label scattering distribution in terms of ascending time order is intuitive given the particular generative procedure of the Mat\'ern type III process. However, working with the global ordering $\sigma$ is impeding when our goal is to discuss the thinned and observed point process separately. We aim for an equivalent form of equation \eqref{labScatOrd} that can be evaluated in any order. We follow \citet{rao2017bayesian} in defining the shadow $\cH(s,t,s^*,t^*)$ of a location $(s^*,t^*) \in \cS \times [0,1]$ at another location $(s,t)$ as the function that takes value 1 if $(s,t)$ would be thinned by $(s^*,t^*)$ and 0 otherwise:
\begin{align}
\cH(s,t,s^*,t^*) = \ind(t>t^*)\ind(||s-s^*||<R). \label{shadow}
\end{align}
Note that a point does not lie in its own shadow. Using this higher level concept, we can now let the innermost products in expression \eqref{labScatOrd} run over all points (instead of only those that precede $i$ in $\sigma$ order) and let the first indicator function in \eqref{shadow} ensure that latter points do not contribute in determining the thinning configuration of earlier points:
\begin{align}
\prod_{i=1}^n \bigg(\bigg[1-\prod_{j=1}^{n} \{1-\cH(s_i,t_i,s_j,t_j)\}^{c_j}\bigg]^{1-c_i}\bigg[\prod_{j=1}^{n} \{1-\cH(s_i,t_i,s_j,t_j)\}^{c_j}\bigg]^{c_i}\bigg). \label{labScat}
\end{align}
Since the thinning indicators are now evaluated for every pair of points, it does not matter in which order we compute each product. That is why we removed any mention of $\sigma$.

As it stands, the procedure for thinning the homogeneous PPP on $\cS \times [0,1]$ operates in a deterministic fashion: we can rule exactly which points are to be labeled as thinned or observed from the locations and times. One of the most interesting contributions of \citet{rao2017bayesian} is to extend the Mat\'ern Type III process to probabilistic thinning. This can be achieved by replacing the second indicator in \eqref{shadow} by a $[0,1]$ valued kernel $K(\cdot;s^*)$ centered at $s^*$. Location $(s,t)$ shall be thinned by $(s^*,t^*)$ with probability
\begin{align*}
\cH(s,t,s^*,t^*) = \ind(t>t^*)K(s;s^*).
\end{align*}
Now each previously unthinned location $s^*$ has a chance to thin a point at $s$, but there is no guarantee even if the two are very close.

We can consider the deterministic version as a special case of the probabilistic one. This generalization does not change the form of the label scattering distribution in \eqref{labScat} and every result in what follows can be expressed in terms of the $[0,1]$ valued shadow function $\cH$, however we choose to define it. Various other extensions can also readily be handled such as location-dependent and/or random radius $R$, but we stay clear of those for simplicity.

In any case, having obtained the expression for the scattering PMF of the thinned ($c=0$) and observed ($c=1$) labels (conditional on locations and times), we can now write the overall density of the base, $\{0,1\}$ marked point process as
\begin{align*}
  f(\{(x_1,t_1, c_1),\dots,(x_n, t_n, c_n)\}) &= p_n \pi_n(x_1,\dots,x_n)\pi_n(t_1,\dots,t_n)\\&\qquad\qquad\qquad\qquad\qquad \pi_n(c_1,\dots,c_n|x_1,t_1,\dots,x_n,t_n)\\
  &= \frac{\exp(-\lambda|\cS|)\lambda^n}{n!} \prod_{i=1}^n\bigg( \bigg[1-\prod_{j=1}^{n} \{1-\cH(s_i,t_i,s_j,t_j)\}^{c_j}\bigg]^{1-c_i}\\
  &\qquad\qquad\qquad\qquad\qquad\qquad\bigg[\prod_{j=1}^{n} \{1-\cH(s_i,t_i,s_j,t_j)\}^{c_j}\bigg]^{c_i}\bigg).
\end{align*}
This is the density of the point process consisting of locations, times and thinning labels.

Now, let $\tilde X_0$ represent the thinned physical locations along with time stamps and let $\tilde X_1$ be its observed counterpart. Together, they form a multitype point process over $\cS \times [0,1]$. From Theorem \ref{colThm}, the joint density of $\tilde X_0, \tilde X_1$ is the density of the marked point process evaluated at the proper thinning labels times a combinatorial factor:
\begin{align}
    f(\tilde S_0, \tilde S_1) &= \frac{\exp(-\lambda |\cS|)\lambda^{n_0+n_1}}{n_0!n_1!}\prod_{i=1}^{n_0}\bigg[1-\prod_{j=1}^{n_1} \{1-\cH(s_{0,i},t_{0,i},s_{1,j},t_{1,j})\}\bigg]\nonumber\\
    &\qquad\qquad\qquad\qquad\qquad\prod_{i=1}^{n_1}\prod_{j=1}^{n_1} \{1-\cH(s_{1,i},t_{1,i},s_{1,j},t_{1,j})\}\nonumber\\
    &\equiv \frac{\exp(-\lambda |\cS|)\lambda^{n_0+n_1}}{n_0!n_1!}\prod_{i=1}^{n_0}h(s_{0,i},t_{0,i};\tilde S_1)\prod_{i=1}^{n_1}\{1-h(s_{1,i},t_{1,i};\tilde S_1)\},\label{jointMT3}
\end{align}
where $\tilde S_k = \{(s_{k,1},t_{k,1}) \dots (s_{k,n_k},t_{k,n_k})\}$ are finite subsets of $\cS \times [0,1]$ ($k=0,1$). In the last equality, we introduce the compact notation $h(s,t;\tilde S_1) = 1-\prod_{j=1}^{n_1} \{1-\cH(s,t,s_{1,j},t_{1,j})\}$. This function maps each location in $\cS \times [0,1]$ to a probability in $[0,1]$ and can be understood as the shadow of all observed locations taken together.

We can compare this joint density with the marginal for the observed points. The likelihood of the observed locations and times was first derived in \citet[Theorem 2.1]{huber2009likelihood} for the Mat\'ern type III process. In their article, they use an accept-reject construction of this process to demonstrate this theorem. We can recover their result by integrating out the thinned point process from the joint density using standard calculations:
\begin{align}
    f(\tilde S_1) &= \int_{\cN_{\text f}(\cS\times [0,1])} f(\tilde S_0,\tilde S_1) \mu_{\text f}(\d \tilde S_0) \nonumber \\
    &= \sum_{n_0\geq 0} \int_{\{\cS\times [0,1]\}^{n_0}} f(\{(s_1,t_1),\dots,(s_{n_0},t_{n_0})\},\tilde S_1)\,\d s_1 \d t_1 \dots \d s_{n_0} \d t_{n_0}\nonumber\\
    &= \frac{\exp(-\lambda |\cS|)\lambda^{n_1}}{n_1!}\prod_{i=1}^{n_1}\{1-h(s_{1,i},t_{1,i};\tilde S_1)\}\sum_{n_0\geq 0} \frac{\lambda^{n_0}}{n_0!} \Big\{\int_{\cS\times [0,1]} h(s,t;\tilde S_1)\, \d s \d t\Big\}^{n_0} \nonumber\\
    &= \frac{\exp(-\lambda \int\{ 1 - h(s,t;\tilde S_1)\}\, \d s \d t)\lambda^{n_1}}{n_1!}\prod_{i=1}^{n_1}\{1-h(s_{1,i},t_{1,i};\tilde S_1)\}. \label{margMT3}
\end{align}
In the last equality, we use the power series representation of the exponential and the fact that $|\cS| = \int \,\d s\d t$. This is reminiscent of a PPP density, but it is not the case as $h(\cdot,\cdot;\tilde S_1)$ is a function of each and every location while PPPs have independent scattering. Just like the PPP density however, this last expression \eqref{margMT3} involves the integral of the shadow function which itself depends on model parameters such as the radius of thinning $R$ in the deterministic version. In contrast, the joint distribution \eqref{jointMT3} is better conditioned for inference.

Now, let us assume that time marks are available at the observed locations in $\cS$ and we wish to instantiate the thinned point process $\tilde X_0$ conditional on $\tilde X_1$ in a data augmentation step. This could be useful for inference schemes where other quantities of interest can be updated elsewhere such as a Gibbs sampling algorithm. In the case of the Mat\'ern type III process, this conditional distribution corresponds to that of a PPP. Indeed, with the observed locations fixed, the thinned locations do not interact with one another as noted in \citet{rao2017bayesian}. This fact can be verified by noting that this conditional density is proportional (as a function of $\tilde S_0$) to the joint of $\tilde X_0,\tilde X_1$:
\begin{align}
    f(\tilde S_0|\tilde S_1) \propto f(\tilde S_0,\tilde S_1) \propto \frac{\lambda^{n_0}}{n_0!}\prod_{i=1}^{n_0}h(s_{0,i},t_{0,i};\tilde S_1). \label{condMT}
\end{align}
With $\tilde S_1$ fixed, we can see that the density in \eqref{condMT} is, up to a normalizing constant, that of a non-homogeneous PPP over $\cS \times [0,1]$  with intensity $\lambda h(s,t;\tilde S_1)$ (recall the form of the PPP density in equation \eqref{PPPCSdens}). One can use a Metropolis-Hastings-based algorithm to sample from this conditional, but a perfect simulation of \eqref{condMT} can be obtained more efficiently by using the thinning procedure of \citet{lewis1979simulation}.

The same result is obtained by using the standard definition $f(\tilde S_0|\tilde S_1) = f(\tilde S_0,\tilde S_1)/f(\tilde S_1)$ which is how \cite{rao2017bayesian} reach the same conclusion. To them however, the joint density of $\tilde S_0, \tilde S_1$ is the same thing as the density of the base, $\{0,1\}$ marked point process. However, this is only true for the particular choice of dominating measure they use (see Remark \ref{rem}) and it is unclear whether the authors are aware of that. For example, the proof they provide would have led to an incoherent result had they used the unit rate PPP distribution as a dominating measure. Nevertheless, the approach we showcase here of working with the joint density gives a principled way of discussing conditional and marginal densities.

\section{Inter-Process Dependence Structure} \label{depsSGCP}

This section explores the pairwise dependence structures that can be induced by the multitype SGCP model introduced in Section \ref{MTSGCP}. We present, for two sets of GP parameter values, the resulting cross pair correlation function. In each case, we showcase how this dependence structure translates in a realization of the point processes and how the MCMC procedure we provide is able to capture this behavior. The relationship between the pair correlation function and the underlying vector valued GP $ g(\cdot)$ is given by
\begin{align}
\gamma_{k\ell}(| s -  t|) = \frac{E[\sigma_k( g( s)) \sigma_\ell ( g( t))]}{E[ \sigma_k( g( s))]E[\sigma_\ell ( g( t))]}, \qquad k,\ell=1,2,\dots,K,  \label{pcf_ratio}
\end{align}
where the transformations are given by 
\begin{align}
\sigma_k( \mathrm g) = P\left(\left\{\argmax ( \mathrm g +  Z) = k \right\} \cap \left\{\max ( \mathrm g +  Z) > 0\right\}\right) \label{transfo}
\end{align}
for a random variable $ Z \sim \cN(0, I)$. The PCF in \eqref{pcf_ratio} is dependent on the GP covariance parameters $ A,\rho_{j=1}^K$ and the GP mean levels $\mu_{j=1}^K$. It is invariant with respect to the base intensity level $\lambda$. A PCF $\gamma_{k\ell}(|\cdot|) > 1$ implies an attractive interaction between processes $k$ and $\ell$ for $k \neq \ell$ whereas $\gamma_{k\ell}(|\cdot|) < 1$ corresponds to a repulsive relationship. Both cases relate directly to the cross covariance function of the positive random fields $\sigma_k( g(\cdot))$ and $\sigma_\ell( g(\cdot))$ with positive (resp. negative) values corresponding to an attractive (resp. repulsive) interaction among processes $k$ and $\ell$.

If we consider the multivariate GP $ g(\cdot)$ to have high mean parameters $\mu_{j=1}^K$, then from equation \eqref{transfo} we have $\sigma_k( \mathrm g) \approx P\left(\argmax ( \mathrm g +  Z) = k\right)$. In that context, there is a sort of competitive interplay between the random fields $\sigma_k( g(\cdot))$ and $\sigma_\ell( g(\cdot))$ and they might exhibit negative dependence even if the individual components $g_k(\cdot)$ and $g_\ell(\cdot)$ are positively correlated. On the other hand, if the mean levels $\mu_{j=1}^K$ are lower, then we have $\sigma_k( \mathrm g) \approx P\left(\mathrm g_k + Z_k > 0\right)$ and in that case the covariance among processes $\sigma_k( g(\cdot))$ and $\sigma_\ell( g(\cdot))$ better reflects the dependence among the underlying processes $g_k(\cdot)$ and $g_\ell(\cdot)$. In other words, lower mean parameters for the $K$-dimensional GP $ g(\cdot)$ leads to a more flexible pairwise dependence structure among the modeled point processes. 

We study two examples that correspond to the attractive and repulsive inter-dependence structures described above. In both cases, we set an equal mean level for the two processes $\mu_1 = \mu_2 = -1$ and employ the same spatial range parameters $\rho_1 = 5$ and $\rho_2 = 10$. In the first example (illustrated in Figure \ref{posStruct}), we induce a positive cross-correlation structure among processes $g_1(\cdot)$ and $g_2(\cdot)$ by employing the coregionalization matrix $ A = [[1,0.5]^\top,[0.5,1]^\top]$ when simulating the multitype SGCP. In the second example (Figure \ref{negStruct}), we showcase a negative dependence structure with $ A = [[1,-0.5]^\top,[0.5,-1]^\top]$. We run the MCMC algorithm described in Section \ref{bayInf} for 10,000 iterations (discarding the first 2,000) for both illustrated realizations and present the 90\% credible intervals around the true cross pair correlation function.

\begin{figure}[h!]
     \centering
     \begin{subfigure}[b]{0.39\textwidth}
         \centering
         \includegraphics[height=2.5in]{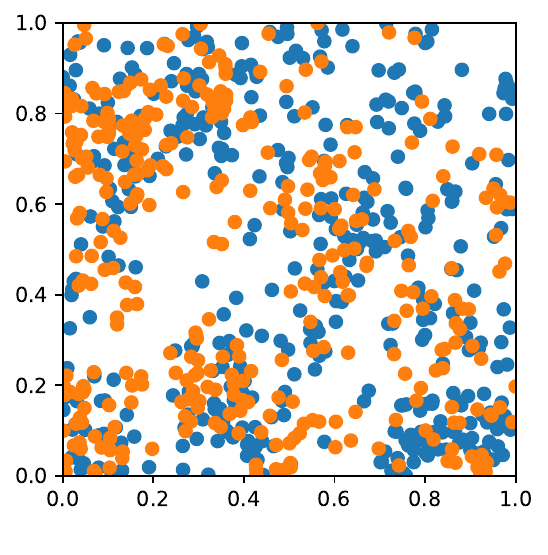}

     \end{subfigure}
    \begin{subfigure}[b]{0.59\textwidth}
        
        \centering
        \includegraphics[height=2.46in]{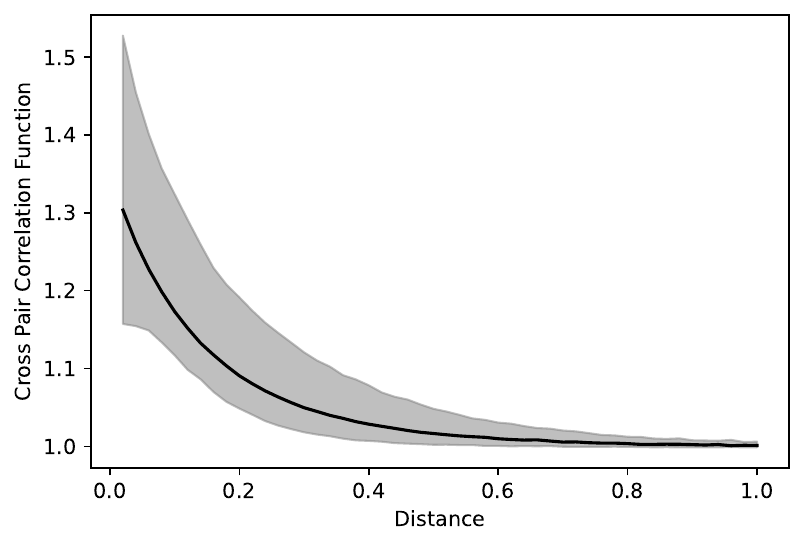}

    \end{subfigure}
        \caption{Positively Correlated Point Processes. On the left: We observe that both processes cluster in the same areas. On the right: The true cross PCF is illustrated in black and the shaded area in gray corresponds to the pointwise 90\% credible intervals computed from posterior samples of the model parameters.}
        \label{posStruct}
\end{figure}

\begin{figure}[h!]
     \centering
     \begin{subfigure}[b]{0.39\textwidth}
         \centering
         \includegraphics[height=2.5in]{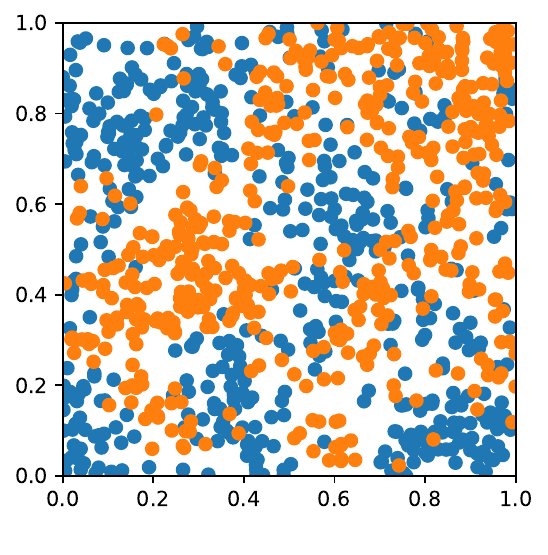}

     \end{subfigure}
    \begin{subfigure}[b]{0.59\textwidth}
        
        \centering
        \includegraphics[height=2.46in]{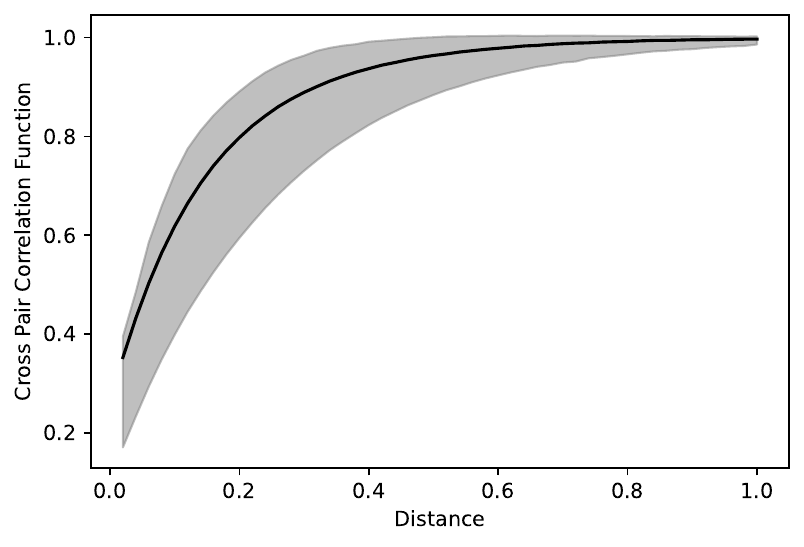}

    \end{subfigure}
        \caption{Negatively Correlated Point Processes. On the left: We observe that both processes cluster in different areas. On the right: The true cross PCF is illustrated in black and the shaded area in gray corresponds to the pointwise 90\% credible intervals computed from posterior samples of the model parameters.}
        \label{negStruct}
\end{figure}

\end{document}